\documentclass[11pt]{article}
\usepackage[letterpaper,margin=1in]{geometry}
\usepackage[utf8]{inputenc}
\usepackage{authblk}
\usepackage{amsmath,amssymb,amsthm}
\usepackage{xspace}
\usepackage{xcolor}
\usepackage{algorithm}
\usepackage{algpseudocode}
\usepackage{tikz}
\usepackage{nicefrac}
\usepackage{pifont}

\newtheorem{theorem}{Theorem}[section]

\newtheorem{corollary}[theorem]{Corollary}
\newtheorem{lemma}[theorem]{Lemma}

\newtheorem{definition}[theorem]{Definition}
\newtheorem{observation}[theorem]{Observation}

\newtheorem{remark}[theorem]{Remark}

\usepackage[small,it]{caption}
\usepackage{enumitem}
\setlist{itemsep=0pt,topsep=0pt}
\newcommand{\Oh}{O}
\newcommand{\prob}[1]{\mathop{\mathrm{Pr} \left[#1\right]}}
\newcommand{\polylog}{\mathop{\mathrm{polylog}}}

\newcommand{\eps}{\varepsilon}
\newcommand{\nochar}{\diamond}

\newfloat{procedure}{ht}{loa}
\floatname{procedure}{Procedure}

\begin{document}
\title{Improved bounds for testing Dyck languages}
\author[1]{Eldar Fischer}
\author[2]{Fr\'{e}d\'{e}ric Magniez}
\author[3]{Tatiana Starikovskaya}
\affil[1]{Technion~- Israel Institute of Technology, Israel\\\texttt{eldar@cs.technion.ac.il}}
\affil[2]{CNRS, IRIF, Univ Paris Diderot, France\\\texttt{magniez@irif.fr}}
\affil[3]{IRIF, Univ Paris Diderot, France\\\texttt{tat.starikovskaya@gmail.com}}

\date{\empty}

\begin{titlepage}
\clearpage\maketitle
\thispagestyle{empty}

\begin{abstract}
In this paper we consider the problem of deciding membership in Dyck languages, a fundamental family of context-free languages, comprised of well-balanced strings of parentheses. In this problem we are given a string of length $n$ in the alphabet of parentheses of $m$ types and must decide if it is well-balanced. We consider this problem in the property testing setting, where one would like to make the decision while querying as few characters of the input as possible.

Property testing of strings for Dyck language membership for $m=1$, with a number of queries independent of the input size $n$, was provided in [Alon, Krivelevich, Newman and Szegedy, SICOMP 2001]. Property testing of strings for Dyck language membership for $m \ge 2$ was first investigated in [Parnas, Ron and Rubinfeld, RSA 2003]. They showed an upper bound and a lower bound for distinguishing strings belonging to the language from strings that are far (in terms of the Hamming distance) from the language, which are respectively (up to polylogarithmic factors) the $2/3$ power and the $1/11$ 
power of the input size $n$.

Here we improve the power of $n$ in both bounds. For the upper bound, we introduce a recursion technique, that together with a refinement of the methods in the original work provides a test for any power of $n$ larger than $2/5$. For the lower bound, we introduce a new problem called Truestring Equivalence, which is easily reducible to the $2$-type Dyck language property testing problem. For this new problem, we show a lower bound of $n$ to the power of $1/5$.
\end{abstract}
\end{titlepage}

\section{Introduction}
\label{sec:intro}
\subsection{Background}
Initially identified as one of the ingredients for the proof of the PCP theorem~\cite{Arora1998}, property testing is nowadays one of the successful paradigms of computation for handling massive data sets. In property testing one would like to decide whether an input has a global property by performing only few local checks. The goal is to distinguish with sufficient confidence the inputs which satisfy the property from those that are far from satisfying it.
In this sense, property testing is a notion of approximation for the corresponding decision problem. Property testers, under the name of self-testers and with a slightly different objective, were first considered for programs computing functions with some algebraic properties~\cite{bk95,blr93,rs96,rub99}.
The notion in its full generality was defined by Goldreich, Goldwasser and Ron, and successfully applied to topics including testing properties of graphs~\cite{ggr98,gr02}, 
monotonicity~\cite{Goldreich2000}, 
group and field operations~\cite{ERGUN2000},
geometrical objects~\cite{Czumaj2000},
formal languages~\cite{D1-test}, 
and probability distributions~\cite{batu2013}. 
The setting of property testing has also been addressed for quantum computers, see~\cite{Mont16} for a survey.

Formally, property testers have random access to their input in the query model. They can read the input, one piece at a time, by submitting a query with the selected index. Ideally, property testers should perform a number of queries that depends only on the approximation parameter (and not on the input length $n$), but also an algorithm making a number of queries that is sublinear in $n$ (for every fixed approximation parameter) is considered a legitimate property test. 
Whereas the complexity of formal language membership is quite well-understood, both in term of space and time complexity (see for instance~\cite{VONBRAUNMUHL1983}), even on parallel architectures~\cite{Dymond2000} and in the streaming model~\cite{Magniez2015}, very little is known in the context of property testing, except when both models of streaming and property testing are combined~\cite{Chu2007,Francois2016}.

The study of property testers for formal languages was initiated by Alon et al.~\cite{D1-test} under the Hamming distance.
In this context, two strings of length $n$ are $\eps$-far if the Hamming distance between them is larger than $\eps n$. An $\eps$-tester for a language $L$ must distinguish, with  probability at least $2/3$, the strings that are in $L$ from those that are $\eps$-far from $L$, using as few queries as possible. Alon et al. showed that regular languages, as well as the Dyck language $D_1$ of well-parenthesized expressions with a single type of parentheses, can be $\eps$-tested with a number of queries independent of the input size $n$. However, by the work of \cite{Dm-test}
the query complexity of an $\eps$-tester for Dyck languages $D_m$ with $m\geq 2$ types of parentheses is between $\Omega(n^{1/11})$ and $\tilde{O}(n^{2/3})$, and by \cite{D1-test} it becomes linear, even for one type of parentheses, when one-sided error is required. When the distance allows sufficient modifications of the input, such as moves of arbitrarily large factors, it has been shown that any context-free language is testable with a constant number of queries~\cite{FMR10}.

\subsection{Motivation and our results}
Dyck languages are not only some of the simplest context free languages, but they are also universal in some sense, since any context-free language can be expressed as an intersection of a regular language with $D_m$, for some integer $m$ and up to some homomorphic map (Chomsky-Sch\"{u}tzenberger representation theorem~\cite{Autebert1997}).
Moreover, Dyck languages have been used in many real life applications over the years, and some of their extensions, such as visibly pushdown languages or nested strings~\cite{AM09}, are heavily used to handle semi-structured documents such as massive databases, or to capture safety properties of programs from their execution traces. 

Motivated by the new applications and the prevalence of massive data, there is a renewed interest in studying the complexity of testing membership in context free languages and in particular in Dyck languages. 
As an illustration, we mention some of the recent works pioneered by Saha
for estimating the edit distance of a string to a Dyck language~\cite{Saha2014} and to context free languages~\cite{Saha2015,Bringmann2016}. 

In this paper we revisit the complexity of property testing for Dyck languages $D_m$, with $m \ge 2$, and improve the previously known upper and lower bounds
(respectively $\tilde{O}(n^{2/3})$ and $\Omega(n^{1/11})$), 
significantly narrowing the gap between the two. Our contribution is twofold:

\begin{enumerate}
\item Our first main result consists of a new property testing algorithm for Dyck languages. In particular, we show that for any $m \ge 2$ and for any constant $\delta,\eps > 0$ there is an $\eps$-tester with complexity $\tilde{O}(n^{2/5+\delta})$, up to some polylogarithmic factors (Corollary~\ref{cor:D_m_upper}).
\item Our second main result is an improved lower bound of $\Omega(n^{1/5})$ (Corollary~\ref{cor:D_m_lower}) for some constant~$\eps$. To show this, we introduce a new formalism which can potentially simplify further establishments of lower bounds in property testing.
\end{enumerate}

\subsection{Overview of the paper}
\paragraph{New algorithms.}
In order to improve the previously known algorithms, we introduce two recursion techniques, that together with a refinement of the methods from previous works provide a better complexity.
We start with an easy reduction from \textsc{$D_m$-membership} testing to \textsc{$D_m$-consistency} testing (Section~\ref{sec:prelim}, Lemma~\ref{lm:red0}), where the later problem simply asks if the input string is (close to) a substring of a member of $D_m$. The reduction is built upon the tester for \textsc{$D_1$-membership} of~\cite{D1-test}, and was implicit in the analysis of the tester of~\cite{Dm-test}. Next, in Section~\ref{sec:D_m_upper}, we show an algorithm for \textsc{$D_m$-consistency} testing. 
The algorithm (Algorithm~\ref{alg:C_m}) partitions the input string into non-overlapping blocks. Note that if the input string is $D_m$-consistent, each of the blocks is $D_m$-consistent as well. A natural part of the algorithm therefore consists of selecting several blocks at random and testing if they are $D_m$-consistent. Instead of querying all characters of the selected blocks as in~\cite{Dm-test}, we design a careful analysis allowing a recursive call here (Theorem~\ref{thm:D_m_consistency}). 

\textit{Inter-block matching.} We call the parentheses of a block that must be matched with parentheses in a different block  \emph{excess}. The crucial part of the algorithm consists of checking that the excess parentheses of the blocks can be matched correctly. In Section~\ref{sec:mg} we construct a dependency graph in order to identify candidate pairs of blocks having many matching excess parentheses between them. 
This graph is an approximation of the \emph{matching graph} introduced in~\cite{Dm-test}
and it can be computed efficiently (Algorithm~\ref{alg:AMG}).
We provide a new property of this graph in the context of $D_m$-consistency testing. 
Namely, the total weight of its edges accounts for almost all excess parentheses of the blocks that are not excess in the input string (Lemma~\ref{lm:total_weight_approx_matching_graph}).
Once such a candidate pair of blocks has been identified, we (approximately) locate their substrings containing the excess parentheses that are required to match. We now need to check whether the excess parentheses of one substring can be matched with the excess parentheses of the second substring. This task is performed by an external algorithm (Algorithm~\ref{alg:substring_matching}), which is itself recursive. We call this task \textsc{Substring $\eps$-matching}. 

\textit{Substring $\eps$-matching.} Our solution for \textsc{Substring $\eps$-matching} exploits a new methodology that combines a recursive approach and the birthday paradox. In order to recurse, we would like to divide each substring into smaller subblocks, and recurse on a random pair of the subblocks that contain matching parentheses. However, since we could only locate the substrings approximately, it is hard to identify efficiently the pairs of subblocks to be tested. We overcome this technical hurdle by trying to guess the real borders and testing each of our guesses. Potentially, this can be very expensive in terms of query complexity. To avoid this, our solution uses the birthday paradox. The idea is not to query a separate subset of subblocks for each pair of tested substrings, but to query a square root number of subblocks from both substrings and then re-use them to test each of our guesses. In the end, we show that Algorithm~\ref{alg:substring_matching} solves \textsc{Substring $\eps$-matching} on $n$-length substrings with query complexity $O(n^{\frac{1}{2}+\delta})$, for any constant $0<\delta\leq 1/2$ (Theorem~\ref{th:substring_matching}).

\paragraph{New lower bound and methodology.} 
For the improved lower bound (Corollary~\ref{cor:D_m_lower}), we first introduce a new problem called \textsc{Truestring equivalence} (see  Section~\ref{sec:prelim}), that highlights one particular aspect of testing Dyck languages. In this problem we must decide if two given binary strings with dummy characters ``$\nochar$'' are equal after deleting all ``$\nochar$'' characters. The ``$\nochar$'' characters hide the indexes of the meaningful bits, just as the indexes of excess parentheses are hidden in a parenthesized expression.
After a quick reduction from \textsc{Truestring equivalence} to \textsc{$D_m$-membership} (Lemma~\ref{lm:reduction}), we proceed to prove a bound for \textsc{Truestring equivalence} using the traditional Yao's method (Theorem~\ref{th:TE_lower})
in Section~\ref{sec:lb}. 
Namely, we produce a distribution $\mathcal{D}_P$ over inputs satisfying \textsc{Truestring equivalence}, and a distribution $\mathcal{D}_N$ over inputs that are (mostly) far from satisfying \textsc{Truestring equivalence} (Lemma~\ref{lem:far}), and then show that any (possibly adaptive) deterministic algorithm will have only a small difference in its acceptance probability when it is fed either an input drawn according to $\mathcal{D}_P$ or one drawn according to $\mathcal{D}_N$ (Lemma~\ref{lem:undistinguishable}).

\textit{A new formalism.} To facilitate the analysis of the behaviour of a deterministic algorithm $\mathcal{A}$, when fed an input drawn according to either distribution, we formalize the technique (at times used informally) of ``revealing unrequested information to the algorithm'' for the purpose of separating the ``irrelevant'' probabilistic dependencies, and allowing a probabilistic analysis that exploits the remaining ``relevant independency'' between the participating random variables.

For this purpose we define the notion of a {\em super-oracle}, which in response to a requested query may output additional information about the input in addition to the query itself (essentially adding additional queries to the requested query and answering all of them). Then, we analyze the distribution over the transcript of any algorithm that is run against the super-oracle, when the input is drawn by either $\mathcal{D}_P$ or $\mathcal{D}_N$.

A special event that we define below is an event in which the super-oracle provides the input in its entirety, essentially ``giving-up'' on confounding the algorithm. For a carefully chosen trigger for this event, we show that the distributions of the algorithm's behaviour (with respect to either $\mathcal{D}_P$ or $\mathcal{D}_N$) both ``underlie'' a common distribution (by which we mean that all outcome probabilities outside the giving-up event are bounded by the respective probabilities of the common distribution), while the probability of the giving-up event itself is small. This implies that the difference between the respective acceptance probabilities is also small, which allows us to invoke Yao's method to conclude the lower bound argument.

\section{Preliminaries}\label{sec:prelim}
\subsection{Basic definitions and some reductions}
Hereafter $n$ will denote the input size, and $\tilde{\Oh}(f(n))$ stands for $\Oh(f(n) \cdot \polylog n)$.
The \emph{distance} between two strings $T,T'$ of length $n$ is the {Hamming distance}, that
is the number of indexes in which they differ. 
We say that $T,T'$ are \emph{$\eps$-close} when 
their distance is at most $\eps n$,
and that they are \emph{$\eps$-far} otherwise.

\begin{definition}[$\eps$-Tester]
Let $L$ be a language over a constant-size alphabet.
A randomized algorithm $A$ is an \emph{$\eps$-tester for $L$}
with bounded error $\eta\geq 0$, 
if $A$ accepts all inputs $T\in L$ with probability at least $1-\eta$, 
and rejects all inputs $T$ that are $\eps$-far from all members of $L$ with probability at least $1-\eta$.
\end{definition}

Usually, the notion of property testing is studied in the context of \emph{query complexity}.
In this model, the algorithm is given the size of the input, but not the input string itself: the algorithm can only access the string by querying it locally, one character at a time.
The query complexity of the algorithm is defined as the number of performed queries. In this work we study the worst case query complexity of the algorithms and disregard  other notions of space and time complexity. 

The Dyck language $D_m$ is the language of strings of properly balanced parentheses of $m$ types. For example, a string ``$( _0(_1 )_1 )_0$'' is in $D_2$, while ``$(_0 (_1 )_0 )_0$'' and ``$(_0 (_1)_1 (_0$'' are not. 
A string $S$ is \emph{$D_m$-consistent}, if it is a substring of a string $S' \in D_m$.
We now define the three main problems that we will consider, and state some reductions between them.
The notion of $\eps$-testing is implicitly extended to those problems by considering the respective languages they define.
\begin{quote}
\textsc{$D_m$-membership$(n)$}\\
Input: String of even length $n$ on an alphabet of parentheses of $m$ types\\
Output: Decide if it is in $D_m$ \smallskip\\
\textsc{$D_m$-consistency$(n)$}\\
Input: String of length $n$ on an alphabet of parentheses of $m$ types\\
Output: Decide if it is $D_m$-consistent
\end{quote}

The last problem is defined in a slightly different but related context.
Given a string $w\in\{0,1,\nochar\}^*$, its {\em truestring} $T(w)$ is the subsequence resulting from deleting all $\nochar$ characters. Two strings $w$ and $v$ are called {\em truestring equivalent} if $T(w)=T(v)$.
\begin{quote}
\textsc{Truestring equivalence$(n)$}\\
Input: Two strings of length $n$ over alphabet $\{0,1,\nochar\}$\\
Output: Decide if they are truestring equivalent
\end{quote}

We will need a tester for \textsc{$D_1$-membership} by Alon et al.~\cite{D1-test}.

\begin{lemma}[\cite{D1-test}]\label{lm:D_1_upper}
There is an $\eps$-tester for \textsc{$D_1$-membership$(n)$} with bounded error $1/6$
and query complexity $\Oh(\eps^{-2}\log (1/\eps))$.
\end{lemma}

\begin{definition}[\cite{Dm-test}]
For a string $S$ on the alphabet of parentheses of $m$ types,  let $\mu(S)$ be a string obtained from $S$ by removing the types of the parentheses.  

Let $k$ and $\ell$ be the smallest integers such that $(^k \mu(S) )^\ell \in D_1$. We call $e_1 (S) = k$ the \emph{excess number of closing parentheses} in $S$, and $e_0(S) = \ell$ the \emph{excess number of opening parentheses} in $S$.

If $S$ is a substring of the input string $T$, its excess numbers indicate how many parentheses cannot be matched with other parentheses in $S$ and must be matched with parentheses outside $S$. Such parentheses are called \emph{excess parentheses} of $S$.  
\end{definition} 

For example, if $S = $ ``$)_0 (_1 )_1 )_0 (_1 )_1$'', then $\mu(S)$ is ``$)())()$'', $e_1(S) = 2$ and $e_0(S) = 0$, and the excess parentheses are the first and the fourth ones. Let $n_0(S')$ be the number of opening parentheses in a substring $S'$ of $S$, and $n_1(S'')$ be the number of closing parentheses in a substring $S''$ of $S$. It is not hard to see that the following equations hold, where we assume that the empty prefix and the empty suffix are included:
\begin{equation}\label{eq:excess}
\begin{split}
e_1(S) =  \max_{S'-\mbox{prefix of }S} \bigl( n_1(S') - n_0(S')\bigr), \quad and \quad
e_0(S) =  \max_{S''-\mbox{suffix of }S} \bigl( n_0(S'') - n_1(S'')\bigr)
\end{split}
\end{equation}

We can now state our reductions.
\begin{lemma}\label{lm:red0}
Given an $\eps$-tester $A$ for \textsc{$D_m$-consistency$(n)$} with bounded error $1/6$, one can design an $\Theta(\eps)$-tester $B$ for \textsc{$D_m$-membership$(n)$} with bounded error $1/3$ and query complexity equal to that of $A$ with an additional term of $\Oh(\eps^{-2}\log (1/\eps))$.
\end{lemma}
\begin{proof}
Our tester for \textsc{$D_m$-membership$(n)$} on $T$ consists of two steps: first, we apply the tester for \textsc{$D_1$-membership$(n)$} of Lemma~\ref{lm:D_1_upper} on $\mu(T)$, and 
then the tester $A$. Observe that $T$ is in $D_m$ if and only if $\mu(T)$ is in $D_1$ and $T$ is $D_m$-consistent, and therefore if $T$ is in $D_m$ it will be accepted with probability at least $ 2/3$.

We now prove by contrapositive that when $T$ is $4\eps$-far from $D_m$ then either $\mu(T)$ is $\eps$-far from $D_1$ or $T$ is $\eps$-far from any $D_m$-consistent string. Suppose that $\mu(T)$ is $\eps$-close to $D_1$ and $T$ is $\eps$-close to a $D_m$-consistent string. First note that since $\mu(T)$ is $\eps$-close to $D_1$, then $T$ contains at most $2\eps n$ excess parentheses. Indeed, if $T$ contains more than $2\eps n$ excess parentheses, then we have to modify at least $\eps n$ of them to obtain a string $\tilde{T}$ such that $\mu(\tilde{T}) \in D_1$, a contradiction. Since $T$ is $\eps$-close to a $D_m$-consistent string, we can modify $\le \eps n$ characters in it so that the resulting string $T'$ is $D_m$-consistent. By modifying $\le \eps n$ characters of $T$ we change its excess numbers by at most $\eps n$ (see Equation~\ref{eq:excess}). 
Therefore, the number of excess parentheses in $T'$ is at most $3 \eps n$. It must be even as well. We change the first half of excess parentheses to ``$(_0$'', and the second half to ``$)_0$'', obtaining a string in $D_m$. 
\end{proof}

\label{sec:reduction}
\begin{lemma}\label{lm:reduction}
Given an $\eps$-tester $A$ for \textsc{$D_2$-membership$(4n)$}, one can design an $\Theta(\eps)$-tester $B$ for \textsc{Truestring equivalence$(n)$} with the same query complexity. % than  $A$.
\end{lemma}
\begin{proof}
Let $w,v \in \{0, 1, \nochar\}^n$. Define $w'$ from $w$ where we replace ``$0$'' by ``$(_0(_0$'', ``$1$'' by ``$(_1(_1$'', and ``$\nochar$'' by ``$(_0)_0$'', and $v'$ from $v$ where we replace  ``$0$'' by ``$)_0)_0$'', ``$1$'' by ``$)_1)_1$'', and ``$\nochar$'' by ``$(_0)_0$''.
We perform the reduction of a pair $(w,v)$ to a string of parentheses $u$ by concatenating $w'$ and the reverse of $v'$. 
It is clear that this maps a pair of truestring equivalent strings to a $4n$-length string in $D_2$, as well as that a query to $u$ can be simulated using a single query to $w$ or $v$. 

We now show that if $u$ is $\eps$-close to $D_2$, then $(w,v)$ is $\Oh(\eps)$-close to a pair of truestring equivalent strings. It suffices to show that we can delete $\Oh(\eps n)$ characters of $T(w)$ and $T(v)$ so that the resulting strings are equal, because we can simulate a deletion from $T(w)$ or $T(v)$ by replacing the corresponding character of $w$ or $v$ with ``$\nochar$''. By definition, there is a string $\tilde{u} \in D_2$ such that the Hamming distance between $u$ and $\tilde{u}$ is $k \le \eps \cdot (4n)$. Moreover, there is a perfect matching on the parentheses of $\tilde{u}$ such that each two matched parentheses $\tilde{u}[i], \tilde{u}[j]$ are of the same type and $|i-j+1|$ is even. We now mark some characters of $\tilde{u}$. Namely, we mark each character $\tilde{u}[i] \neq u[i]$ and its matching parenthesis. Also, if $\tilde{u}[i]$ was marked and $u[i-1, i]$ or $u[i, i+1]$ was obtained by replacing a ``$\nochar$'' character with the sequence ``$(_0)_0$'' in $w$ or $v$, we mark $\tilde{u}[i-1]$ or $\tilde{u}[i+1]$ respectively, as well as its matching parenthesis in $\tilde{u}$ (some such characters might have been already marked before). Finally, we mark all untouched pairs of ``$(_0)_0$'' corresponding to ``$\nochar$'' characters.

Consider a character of $T(w)$ or $T(v)$ and the corresponding sequence $\tilde{u}[i, i+1]$.
If both $\tilde{u}[i]$ and $\tilde{u}[i+1]$ are marked, we delete the character. In total, we delete $\Oh(k) = \Oh(\eps n)$ non-``$\nochar$'' characters. To show that the resulting strings are equal, note that the set of unmarked characters of $\tilde{u}$ is comprised of matching parentheses (because each time we marked a pair of matching parentheses), and contains only those characters where $u$ and $\tilde{u}$ agree. Moreover, each unmarked character $\tilde{u}[i]$, $i\leq 2n$, is an opening parenthesis that matches some unmarked closing parenthesis $\tilde{u}[j]$, $j>2n$, where $|j-i+1|$ is even. 
\end{proof}

\subsection{Excess parentheses preprocessing}
\label{sec:excess}
Parnas et al.~\cite{Dm-test} showed that it suffices to query $\tilde{\Oh}(n^{2/3}/\eps^2)$ indexes of the input string $T$ to compute the excess numbers of any substring of $T$ of length $\geq n^{2/3}$ with precision $\eps n^{2/3}$. Below we show a new approach that will allow us to approximate excess numbers of any substring independent of its length, which is important for our recursive tester. From Equation~\ref{eq:excess} it follows that to estimate the excess numbers it suffices to estimate the number of opening and closing parentheses in each prefix and suffix of~$S$. 

\begin{lemma}\label{lm:HN_all_substrings}
By querying $\tilde{\Oh}(x^2 /\Delta^2)$ indexes of a string $S'$ of length $x \le n$, there is an algorithm computing the number of opening or closing parentheses in any substring $S$ of $S'$ with precision $\Delta$ correctly with probability at least $1-1/n^3$.
\end{lemma}
\begin{proof}
We query a subset of $(2 x^2 \log n) /\Delta^2$ indexes of $S'$ uniformly at random. For each substring $S'$ of length $\le \Delta$ we can output $\Delta$ as an approximation of the number of opening or closing parentheses. Consider now any substring $S$ of $S'$ of length $y \cdot \Delta$, where $1 < y \le x / \Delta$. By Chebyshev's inequality, it contains $\ge y \cdot (x \log n /\Delta)$ queried indexes with probability $\ge 1/2$. We repeat this step $\log (2n^3)$ times to amplify the probability. As a corollary, $S$ will contain $\ge y \cdot (x \log n /\Delta)$ queried indexes with probability $\ge 1-1/2n^3$. We divide the samples into $\log n$ subsets of size $y \cdot (x/\Delta)$. Consider one such subset of indexes $p_1, \ldots, p_{y \cdot (x/\Delta)}$. Setting $X_i = 1$ if $S[p_i]$ is an opening parenthesis and $X_i=0$ otherwise, for $X=\sum_{i=1}^{y \cdot (x/\Delta)} X_i$ we have $\mathbb{E}[X] = n_0 \cdot \frac{y \cdot (x/\Delta)}{y \cdot \Delta}$. By the additive Chernoff bound we then obtain 

$$\prob{|X - n_0 \cdot \frac{y \cdot (x/\Delta)}{y \cdot \Delta}| \ge \sqrt{y \cdot (x/\Delta)}} \le 2e^{-2} < 1/3$$

Dividing the inequality under the probability by $x/\Delta^2$ we obtain 

$$\prob{|X\cdot (\Delta^2/x) - n_0| \ge \frac{\sqrt{y \cdot (x/\Delta)}}{x/\Delta^2}} \le 2e^{-2} < 1/3$$

Since $\frac{\sqrt{y \cdot (x/\Delta)}}{x/\Delta^2} \le \Delta$ (recall that $y \le x/\Delta$), we obtain that $\hat{n}_0 = X\cdot (\Delta^2/x)$ is a $\Delta$-approximation of $n_0$ with probability $> 2/3$. We amplify the probability by taking the median of the values computed over all subsets of indexes.
\end{proof}

\begin{lemma}\label{lm:excess_approx}
By querying $\tilde{\Oh}(x^2 /\Delta^2)$ random indexes of a string $S'$ of length $x \le n$,
there is an algorithm computing the excess numbers of any substring $S$ of $S'$ with precision $\Delta$ correctly with probability at least $1-1/n^3$.
\end{lemma}
\begin{proof}
The lemma follows immediately from Equation~\ref{eq:excess} and Lemma~\ref{lm:HN_all_substrings} for $\Delta = \Delta / 2$.
\end{proof}

\subsection{Matching graph} \label{sec:mg}
Let us first remind the notion of a matching graph introduced by Parnas et al.~\cite{Dm-test}. 
Let $k, \ell$ be the excess numbers of $T$, i.e. the smallest integers such that $T' = (^k \; \mu(S) \; )^\ell \in D_1$. Since $T' \in D_1$, there is a unique perfect matching on its characters. Let $b = n^{4/5}$. We divide $T$ into non-overlapping blocks of length $b$ (the last block may be shorter).

\begin{definition}[Matching graph]\label{def:matching_graph}
The matching graph $G = (V,E)$ of $T$ is a weighted graph where $V$ is a set of the blocks of $T$. If $w(i,j)$ parentheses in block $i$ match parentheses in block $j$, then the two blocks $i, j$ are connected by an edge $(i,j)$ of weight $w(i,j)$.
\end{definition}

In other words, the matching graph tells if the blocks $i,j$ contain matching parentheses, and also the number of such parentheses. Compared to the definition given in~\cite{Dm-test}, we changed the size of the blocks, which will allow us to use recursion and to improve the upper bound. This change does not affect the properties of the matching graph stated in~\cite{Dm-test}.

\begin{remark}[\cite{Dm-test}]
\label{rm:matching_graph_planarity}
The matching graph $G$ is planar and therefore has at most $3 n/b$ edges.
\end{remark}

We say that blocks $i$ and $j$ are neighbours if there is an edge between them. Let $T_{i,j}$ be the substring of $T$ containing blocks $i$-th to $j$-th inclusively. 

\begin{lemma}[\cite{Dm-test}]\label{cor:weight}
Let $i \neq j$ be two blocks of $T$ and define $\sigma (i,j) = \min\{e_0(T_{i,i}), e_1(T_{i+1,j})\} - e_1(T_{i+1,j-1})$. The following is true: (a) If $\sigma (i,j) > 0$, then $i,j$ are neighbours; (b) If $i, j$ are neighbours, $w(i,j) = \sigma (i,j)$. 
\end{lemma} 

We will compute the matching graph approximately using Algorithm~\ref{alg:AMG}.
It relies on the approximation of excess parentheses with precision $\eps b$ from Lemma~\ref{lm:excess_approx}. 
We call the resulting output graph the \emph{approximate matching graph $\hat{G}$}.

\begin{algorithm}
\caption{Approximate matching graph $\hat{G}$}
\label{alg:AMG}
Input: string $T$ of size $n$
\begin{enumerate}
\item  Divide $T$ into non-overlapping blocks of length $b= n^{4/5}$ 
\item Run the excess parentheses preprocessing for precision $\eps b$ (Lemma~\ref{lm:excess_approx})
\item For each $i, j \in \{1, \ldots, n/b\}$, $i \neq j$:
\begin{enumerate}
\item Get $\hat{e}_0(T_{i,i})$, $\hat{e}_1(T_{i+1,j})$, and $\hat{e}_1(T_{i+1,j-1})$
\item Compute $\hat{\sigma}(i,j) = \min\{\hat{e}_0(T_{i,i}), \hat{e}_1(T_{i+1,j})\} - \hat{e}_1(T_{i+1,j-1})$
\end{enumerate}
\item Construct the weighted graph $\hat{G} = (V, \hat{E})$ where $V$ is a set of blocks of $T$ 
and $\hat{E}$ is the set of edges $(i,j)$ such that $\hat{\sigma}(i,j) > 8 \eps b$, with respective weights
$\hat{w}(i,j)=\hat{\sigma}(i,j)$
\end{enumerate}
\end{algorithm}

The approximate matching graph satisfies the following property.

\begin{lemma}[\cite{Dm-test}]\label{lm:approx_matching_graph}
With probability at least $1 - 1/n$, the approximate matching graph $\hat{G}$ is a subgraph of the matching graph $G$, and every vertex in $\hat{G}$ has degree $\Oh(1/\eps)$. 
\end{lemma}

We also show a new property that will be essential for the analysis of our $D_m$-consistency tester. 
For this, define a \emph{locally excess parenthesis} to be an excess parenthesis of some block of $T$ which is not excess in $T$. We will show that the total weight of the edges of the approximate matching graph accounts for almost all locally excess parentheses.

\begin{lemma}\label{lm:total_weight_approx_matching_graph}
$$\sum_{(i,j) \in \hat{E}: i < j} \hat{w}(i,j) \ge \frac{1}{2} \sum_{i=1}^{i = n/b} \left(e_0(T_{i,i}) + e_1(T_{i,i}) \right) - (e_0(T) + e_1(T)) - \Oh(\eps n).$$
\end{lemma}
\begin{proof}
Consider an edge $(i,j)$ of weight $w(i,j) \ge 9 \eps b$. We then have $\sigma (i,j)\ge 9\eps b$. Consequently $\hat{\sigma}(i,j) >8\eps b$, which implies that $(i,j) \in \hat{E}$. In other words, $(i,j)$ is an edge of $\hat{G}$ as well. By Remark~\ref{rm:matching_graph_planarity}, the total weight of edges $(i,j) \in E$ such that $w(i,j) < 9\eps b$ is at most $3(n/b) \cdot 9\eps b = 27 \eps n$. Therefore,
$$\sum_{(i,j) \in \hat{E}: i < j} \hat{w}(i,j)  = \sum_{(i,j) \in \hat{E}: i < j} \hat{\sigma}(i,j) \ge \sum_{(i,j) \in \hat{E}: i < j} \sigma(i,j) - \Oh(\eps n) = \sum_{(i,j) \in E: i < j} w(i,j) - \Oh(\eps n)$$
Since $T$ contains $e_0 (T) + e_1 (T)$ excess parentheses, we have that
$$\sum_{(i,j) \in E: i < j} w(i,j) \ge \frac{1}{2}\left( \sum_{i=1}^{i=n/b} e_0(T_{i,i})+e_1(T_{i,i})\right) - (e_0(T) + e_1(T))$$
Combining the two inequalities we obtain the claim.
\end{proof}

\section{A tester for $D_m$-consistency}\label{sec:D_m_upper}
Let $T$ be the input string partitioned into non-overlapping blocks of length $b = n^{4/5}$ (except for the last block that may be shorter). Our $D_m$-consistency test consists of two steps: first we check that the excess parentheses of the blocks can be matched correctly, and then (recursively) check that the blocks are $D_m$-consistent. The structure of the test repeats the structure of the test by Parnas et al.~\cite{Dm-test}, 
in particular, both tests are based on the notion of approximate matching graph. However, our test uses new and much more sophisticated techniques, which finally gives us a better bound.

\begin{algorithm}
\caption{\textsc{$D_m$-consistency$(n,k)$}}
\label{alg:C_m}
Input: string $T$ of length $n$
\begin{enumerate}
\item If $k = 0$, run the deterministic $D_m$-consistency $0$-tester, and stop
\item Divide $T$ into non-overlapping blocks of length $b = n^{4/5}$
\item Inter-block matching: 
\begin{enumerate}
\item Compute the approximate matching graph $\hat{G}$ for the blocks using Algorithm~\ref{alg:AMG}
\item Select $\eps^{-1} \log n$ blocks uniformly at random and for each find all of its $\Oh(1/\eps)$  neighbours in $\hat{G}$
\item For each selected block $i$ and for each of its  neighbours $j$: 
\begin{enumerate}
\item Find (approximately) the smallest substring $S_1$ of block $i$ that contains all excess opening parentheses that match in block $j$, and the smallest substring $S_2$ of block $j$ that contains all excess closing parentheses that match in block $i$ (see Section~\ref{sec:global})
\item Check that $S_1,S_2$ $\eps$-match using \textsc{Substring $\eps$-matching$(b)$} (Theorem~\ref{th:substring_matching})
\end{enumerate}
\end{enumerate}
\item $D_m$-consistency of the blocks:
\begin{enumerate}
\item Select $4 \eps^{-1} \log n$ blocks uniformly at random
\item Run the \textsc{$D_m$-consistency$(b,k-1)$} test twice for each of the selected blocks
\end{enumerate}
\end{enumerate}
\end{algorithm}

Algorithm~\ref{alg:C_m} shows the main steps of our new tester \textsc{$D_m$-consistency$(n,k)$}. In Step 3 we call a subroutine \textsc{Substring $\eps$-matching$(x)$}. It must accept $S_1$ if almost all excess parentheses in it can be matched in $S_2$. Formally,

\begin{definition}[sequentially match, $\eps$-match] \label{def:eps_match}
Consider two substrings $S_1, S_2$ of $T$ of maximal length $x$.
We say that the excess opening parentheses of $S_1$ can be \emph{matched sequentially} in $S_2$
 if there is a perfect matching between the excess opening parentheses of $S_1$ and a (continuous) subrange of excess closing parentheses of $S_2$ such that any two matched parentheses $T[i], T[j]$ have the same type and the \emph{distance between them}, defined as $|j-i+1|$, is even. 

We say that $S_1,S_2$ \emph{$\eps$-match} if all but at most  $7\eps x$ leftmost and $5\eps x$ rightmost 
excess opening parentheses of $S_1$ can be matched sequentially in $S_2$ . 
\end{definition}

\begin{quote}
\textsc{Substring $\eps$-matching$(x)$}\\
Input: Two substrings $S_1,S_2$ of $T$ of maximal size $x$\\
Output: Accept if they $\eps$-match;
Reject if at most $ e_0 (S_1) - 30\eps x$ excess opening parentheses of $S_1$ can be matched sequentially in $S_2$.
\end{quote}
The choice of constants is important for the analysis of \textsc{$D_m$-consistency$(n,k)$}.  
In Section~\ref{sec:substring_matching} we show the following theorem by using recursion. 

\begin{theorem}\label{th:substring_matching}
For every constant $0 < \delta \le 1/2$ there is an algorithm for
\textsc{Substring $\eps$-matching$(x)$}
with query complexity $\tilde{\Oh}(\eps^{(2\delta)^{(1/\log 3/4) - 2}+4} x^{1/2+\delta})$ and bounded error $1/n^2$.
\end{theorem}

In Section~\ref{sec:global} we give a detailed description of Algorithm~\ref{alg:C_m}, Step 3 (inter-block matching), and then in Section~\ref{sec:D_m_consistency} prove the following theorem. 

\begin{theorem}\label{thm:D_m_consistency}
For any $0 < \delta < 1/2$,
\textsc{$D_m$-consistency$(n,5)$} (Algorithm~\ref{alg:C_m}) is an $\eps$-tester for \textsc{$D_m$-consistency$(n)$} with query complexity $\tilde{\Oh}(\eps^{(2 \delta)^{(1/\log 3/4) - 2} + 2} n^{2/5+\nicefrac{4}{5} \cdot \delta})$ and bounded error $1/6$.
\end{theorem}

This theorem immediately implies a new tester for $D_m$-membership via Lemma~\ref{lm:red0}.
\begin{corollary}\label{cor:D_m_upper}
For any $0 < \delta < 1/2$ there is an $\eps$-tester for \textsc{$D_m$-membership$(n)$}  with query complexity $\tilde{\Oh}(\eps^{(2 \delta)^{(1/\log 3/4) - 2} + 2} n^{2/5+\nicefrac{4}{5} \cdot \delta})$. 
\end{corollary}

\subsection{Inter-block matching}~\label{sec:global}
In this section we give a detailed description and analyse Step 3 of the \textsc{$D_m$-consistency$(n,k)$} algorithm (inter-block matching). We start by running the excess parentheses preprocessing (Lemma~\ref{lm:excess_approx}) and computing the approximate matching graph for the blocks. Next, we select $\eps^{-1} \log n$ blocks uniformly at random and for each find all of its $\Oh(1/\eps)$ neighbours in the approximate matching graph. 

Consider one of the selected blocks, $i$ and its  neighbour $j$. Assume $i < j$, the other case is symmetrical.
Let $[p,q]$ be the smallest interval of indexes in the block $i$ containing all excess opening parentheses that match in the block $j$, and $[r, s]$ the smallest interval of indexes in the block $j$ containing all excess closing parentheses that match in block~$i$. Unfortunately, we cannot compute the precise values of $p, q, r, s$, but we can compute their approximate values $\hat{p}, \hat{q}, \hat{r}, \hat{s}$ (see Figure~\ref{fig:substring_matching}). We run the \textsc{Substring $\eps$-matching} algorithm (Theorem~\ref{th:substring_matching}) on $T[\hat{p}, \hat{q}]$ and $T[\hat{r}, \hat{s}]$. If the algorithm confirms that the substrings $\eps$-match, we accept $i$ and $j$, and otherwise we reject them. The inter-block matching accepts $T$ if and only if all tested block pairs are accepted. 

\begin{figure}[h!]
\begin{center}
\begin{tikzpicture}[scale=0.6]
\draw[|-|] (0,0)--(10,0);
\node at (5,-2) {Block $i$};
\draw (3,-0.5) rectangle (7,0.7);
\node[below] at (3,-0.4) {$\hat{p}$};
\node[below] at (4,-0.5725) {$p$};
\node[below] at (6,-0.5725) {$q$};
\node[below] at (7,-0.4) {$\hat{q}$};

\node at (3.2,0) {\small{(}};
\node at (3.7,0) {\small{(}};
\node at (4.5,0) {\small{(}};
\node at (5,0) {\small{(}};
\node at (5.2,0) {\small{(}};
\node at (5.6,0) {\small{(}};
\node at (6,0) {\small{(}};

\node at (6.4,0) {\textcolor{red}{\small{(}}};
\node at (6.45,.5) {\textcolor{red}{\tiny{\ding{54}}}};
\node at (6.7,0) {\textcolor{red}{\small{(}}};
\node at (6.75,.5) {\textcolor{red}{\tiny{\ding{54}}}};
\node at (7.3,0) {\textcolor{red}{\small{)}}};
\node at (7.25,.5) {\textcolor{red}{\tiny{\ding{54}}}};
\node at (7.6,0) {\textcolor{red}{\small{)}}};
\node at (7.55,.5) {\textcolor{red}{\tiny{\ding{54}}}};

\draw[|-|] (15,0)--(25,0);
\node at (20,-2) {Block $j$};
\draw (17,-0.5) rectangle (24,0.7);
\node[below] at (17,-0.4) {$\hat{r}$};
\node[below] at (19,-0.5725) {$r$};
\node[below] at (22,-0.5725) {$s$};
\node[below] at (24,-0.4) {$\hat{s}$};

\node at (16,0) {\textcolor{red}{\small{(}}};
\node at (16.05,.5) {\textcolor{red}{\tiny{\ding{54}}}};
\node at (16.3,0) {\textcolor{red}{\small{(}}};
\node at (16.35,.5) {\textcolor{red}{\tiny{\ding{54}}}};
\node at (16.7,0) {\textcolor{red}{\small{(}}};
\node at (16.75,.5) {\textcolor{red}{\tiny{\ding{54}}}};
\node at (17.5,0) {\textcolor{red}{\small{)}}};
\node at (17.45,.5) {\textcolor{red}{\tiny{\ding{54}}}};
\node at (18,0) {\textcolor{red}{\small{)}}};
\node at (17.95,.5) {\textcolor{red}{\tiny{\ding{54}}}};
\node at (18.3,0) {\textcolor{red}{\small{)}}};
\node at (18.25,.5) {\textcolor{red}{\tiny{\ding{54}}}};
\node at (18.6,0) {\small{)}};
\node at (19.2,0) {\small{)}};
\node at (20,0) {\small{)}};
\node at (21.5,0) {\small{)}};
\node at (22,0) {\small{)}};
\node at (23,0) {\small{)}};
\node at (23.3,0) {\small{)}};
\node at (23.6,0) {\small{)}};

\draw (10,0)--(11.5,0);
\draw (15,0)--(13.5,0);
\node at (12.5,0) {\ldots};
\node at (12.5,-2) {$T_{i+1,j-1}$};
\end{tikzpicture}
\end{center}
\caption{Blocks $i$ and $j$, and intervals $[\hat{p}, \hat{q}]$ and $[\hat{r}, \hat{s}]$. Black parentheses are excess parentheses of the blocks $i$ and $j$, red parentheses (also marked by crosses \textcolor{red}{\ding{54}}) are excess in $T[\hat{p}, \hat{q}]$ and $T[\hat{r}, \hat{s}]$, but not in the blocks $i,j$ (they will be matched with the red parentheses just outside of the intervals.)}
\label{fig:substring_matching}
\end{figure}

We now explain how we approximate $p, q, r, s$. Recall that $\hat{w}(i,j) > 8 \eps b$ is the approximate weight of the edge $(i,j)$ and $\hat{e}_0 (S)$, $\hat{e}_1 (S)$ are the excess numbers of opening / closing parentheses in a substring $S$ of $T$ computed with precision $\eps b$. We compute $\hat{p}, \hat{q}, \hat{r}, \hat{s}$ in the following way:
\begin{enumerate}
\item Let $\hat{q}$ be the rightmost index in block $i$ such that $\hat{e}_0(T[\hat{q}, ib]) \ge \hat{e}_1(T_{i+1,j-1}) - 2 \eps b$ and $\hat{e}_1 (T[\hat{q}, ib]) \le \eps b$; 
\item Let $\hat{p}$ be the rightmost index in block $i$ such that $\hat{e}_0 (T[\hat{p}, \hat{q}]) \ge \hat{w}(i,j) + 6 \eps b$ (if there is none, we put $\hat{p} = (i-1)b+1$);
\item Let $\hat{r}$ be the leftmost index in block $j$ such that $\hat{e}_1 (T[(j-1)b+1, \hat{r}]) \ge \hat{e}_0(T_{i+1,j-1}) - 2 \eps b$ and $\hat{e}_0 (T[(j-1)b+1,\hat{r}]) \le \eps b$; 
\item Let $\hat{s}$ be the leftmost index in block $j$ such that $\hat{e}_1(T[\hat{r}, \hat{s}]) \ge \hat{w}(i,j) + 6 \eps b$ (if there is none, we put $\hat{s} = jb$).
\end{enumerate}
Since $\hat{w}(i,j) = \hat{\sigma}(i,j) = \min\{\hat{e}_0(T_{i,i}), \hat{e}_1(T_{i+1,j})\} - \hat{e}_1(T_{i+1,j-1}) > 8 \eps b$, indexes $\hat{q}$ and $\hat{r}$ always exist.

\paragraph{Correctness of inter-block matching.} We now show that if $T$ is $D_m$-consistent, inter-block matching will accept it with high probability, and that if $T$ is accepted, then almost all locally excess parentheses of $T$ can be matched correctly. Recall that a parenthesis of $T$ is called locally excess if it is excess for some block of $T$, but not for $T$ itself.

\begin{lemma}
If $T$ is $D_m$-consistent, it is accepted by inter-block matching with probability $> 1 - 1/n$.
\end{lemma}
\begin{proof}
We will show that if $T$ is $D_m$-consistent, then for every pair of  neighbours $i,j$ in $\hat{G}$ the substrings $T[\hat{p}, \hat{q}]$ and $T[\hat{r}, \hat{s}]$ defined as above $\eps$-match. From Theorem~\ref{th:substring_matching} and the union bound it will immediately follow that $T$ is accepted with probability $> 1 - 1/n$. 

We need to show that all but at most $7 \eps b$ leftmost and at most $5 \eps b$ rightmost excess opening parentheses in $T[\hat{p}, \hat{q}]$ can be sequentially matched in $T[\hat{r}, \hat{s}]$. By definition, all excess parentheses of $T[p,q]$ can be sequentially matched in $T[r,s]$. $T[\hat{r}, \hat{s}]$ contains $T[r,s]$ as a subinterval, and therefore it suffices to show that $T[\hat{p}, \hat{q}]$ has at most $5 \eps b$ extra excess opening parentheses on the right and at most $7 \eps b$ extra excess opening parentheses on the left compared to $T[p,q]$. 

If $T$ is $D_m$-consistent, all excess closing parentheses in $T_{i+1,j-1}$ must match in $T[q+1, ib]$. Therefore, $T[q, ib]$ must contain $e_1 (T_{i+1,j-1})$ excess opening parentheses. From the definition it follows that $T[\hat{q}, ib]$ contains at least $e_1 (T_{i+1,j-1}) - 3 \eps b$ parentheses. It means that $T[\hat{p}, \hat{q}]$ has at most $3 \eps b$ extra excess opening parentheses on the right that must be matched in $T_{i+1,j-1}$. There also can be at most $2\eps b$ extra excess opening parentheses that were not excess in the block $i$ but became excess in $T[\hat{p}, \hat{q}]$ (see Figure~\ref{fig:substring_matching}). In total, there will be at most $5 \eps b$ extra excess opening parentheses on the right.

Consider now the rightmost excess opening parenthesis of $T[\hat{p}, \hat{q}]$ that matches in block $j$. Starting from it, there must be $w(i,j)$ more excess opening parentheses that also match in block $j$. We defined $\hat{p}$ to be the rightmost index in block $i$ such that $\hat{e}_0 (T[\hat{p}, \hat{q}]) \ge \hat{w}(i,j) + 6 \eps b$. It follows that $\hat{e}_0 (T[\hat{p}+1, \hat{q}]) < \hat{w}(i,j) + 6 \eps b < w(i,j) + 7 \eps b$ or that $e_0 (T[\hat{p}, \hat{q}]) \le w(i,j) + 7 \eps b$, which concludes the proof.
\end{proof}

\begin{lemma}\label{lm:global_tests_accepts}
If $T$ is accepted with probability $> 1/n$, then there is a matching on its locally excess parentheses such that: (a) Any two matched parentheses have the same type and the distance between them is even; (b) There  are at most $\Oh(\eps n)$ unmatched locally excess parentheses.
\end{lemma}
\begin{proof}
By Lemma~\ref{lm:excess_approx}, excess preprocessing is correct for all substrings of $T$ with probability $>1 - 1/2n$. We can therefore assume that both assumptions ($T$ is accepted, excess preprocessing is correct) hold with probability $> 1/2n$. 

The fraction of blocks $i$, which have a neighbour $j$ in $\hat{G}$ such that the substring matching test accepts with probability $< 1/n^2$ if executed on $i$, $j$, is at most $\eps$ (otherwise we would select one of them with probability $\ge 1 - 1/n^2$). Let $\mathcal{R}$ be a set of all such blocks. We thus obtain that 
$$\sum_{(i,j) \in \hat{E}: i \neq j, i \in \mathcal{R}} \hat{\sigma}(i,j) = \Oh(\eps n)$$
Consider now any three blocks $i_1 < i_2 < j$ (the case $j < i_1 < i_2$ is analogous.) such that both $i_1$ and $i_2$ are  neighbours of $j$ in $\hat{G}$. Suppose that both $i_1, j$ and $i_2, j$ are accepted by the \textsc{Substring $\eps$-matching} algorithm with probability $> 1/n^2$. Theorem~\ref{th:substring_matching} implies that there is a subsequence of locally excess opening parentheses of block $i_1$ of length $\ge \sigma(i_1,j) - 30 \eps b$ that can be matched with a subsequence $\pi_1$ of locally excess closing parentheses of block $j$, and a subsequence of locally excess opening parentheses of block $i_2$ $\ge \sigma(i_1,j) - 30 \eps b$ that can be matched with a subsequence $\pi_2$ of locally excess closing parentheses of block $j$. The matchings satisfy property (a); moreover, assuming that the excess numbers of all substrings of $T$ were computed correctly with precision~$\eps b$, the subsequences $\pi_1$ and $\pi_2$ overlap by $\le 12 \eps b$ excess parentheses. The latter follows from the definition of $\hat{p}, \hat{q}, \hat{r}, \hat{s}$. Indeed, the number of excess closing parentheses between $(j-1)b$ and the last parenthesis of $\pi_2$ is at most $\hat{e}_0 (T_{i_2+1,j-1}) +\hat{w}(i_2,j) + 7 \eps b \le e_0(T_{i_2+1,j-1})+\sigma (i_2,j)+ 9 \cdot  \eps b$, where $\sigma (i_2,j)=\min\{e_0(T_{i_2,i_2}), e_1(T_{i_2+1,j})\} - e_1(T_{i_2+1,j-1}) \le e_0(T_{i_2,i_2}) - e_1(T_{i_2+1,j-1})$. On the other hand, the number of excess closing parentheses between $(j-1)b$ and the first parenthesis of $\pi_1$ is at least $e_0(T_{i_1,j-1})-3\eps b$. Note that $e_0(T_{i_1,j-1}) \ge e_0(T_{i_2,i_2}) - e_1(T_{i_2+1,j-1}) + e_0(T_{i_2+1,j-1})$. Therefore, the two subsequences overlap by at most $12 \eps b$ excess parentheses. 

From above it follows that the total number of locally excess parentheses in the matched subsequences is $\sum_{(i,j) \in \hat{E}: i < j} \hat{\sigma}(i,j) - \Oh(\eps n)$.  Moreover, each two subsequences overlap by at most $12 \eps b$ parentheses. Since by Lemma~\ref{lm:approx_matching_graph} $\hat{G}$ is a subgraph of $G$, that is planar and therefore has at most $3 n / b$ edges, the total lengths of overlaps is $\Oh(\eps n)$. Therefore, we will be able to match at least $\sum_{(i,j) \in \hat{E}: i < j} \hat{\sigma}(i,j) - \Oh(\eps n)$ locally excess parentheses. The claim follows from Lemma~\ref{lm:total_weight_approx_matching_graph}.
\end{proof}

\paragraph{Complexity of inter-block matching.} To compute the approximate matching graph we need $\tilde{\Oh}(\eps^{-2} n^{2/5})$ queries. The substring matching test is called $\tilde{\Oh}(\eps^{-2} \log n)$ times, and takes $\tilde{\Oh}(\eps^{(2 \delta)^{1/\log 3/4 - 2} + 4} b^{1/2+\delta})$ queries for a fixed constant $0 < \delta < 1/2$. Since $b = n^{4/5}$, we finally obtain that inter-block matching can be implemented to have complexity $\tilde{\Oh}(\eps^{(2 \delta)^{1/\log 3/4 - 2} + 2} n^{2/5+\nicefrac{4}{5} \cdot \delta})$ for any $0 < \delta < 1/2$.

\subsection{Recursion (proof of Theorem~\ref{thm:D_m_consistency})}\label{sec:D_m_consistency}
We are now ready to prove Theorem~\ref{thm:D_m_consistency} that claims that \textsc{$D_m$-consistency$(n,5)$} test is an $\eps$-tester for \textsc{$D_m$-consistency$(n)$} with complexity $\tilde{\Oh}(\eps^{(2 \delta)^{1/\log 3/4 - 2} + 2} n^{2/5+\nicefrac{4}{5} \cdot \delta})$. 

We start by analysing the complexity of the test. The pseudocode of the test (see Algorithm~\ref{alg:C_m}) directly implies that if \textsc{$D_m$-consistency$(n,k-1)$} test has query complexity $f(n, \eps)$, then \textsc{$D_m$-consistency$(n,k)$} test has query complexity $\tilde{\Oh}(\eps^{-1} f(n^{4/5}, \eps) +\eps^{(2 \delta)^{1/\log 3/4 - 2} + 2} n^{2/5+\nicefrac{4}{5} \cdot \delta})$, where $0 < \delta < 1/2$ is a constant in the complexity of \textsc{Substring $\eps$-matching} (Theorem~\ref{th:substring_matching}). Recall that for the base case $k = 0$, the \textsc{$D_m$-consistency$(n,0)$} test is the trivial deterministic $0$-tester for $D_m$-consistency with query complexity $f(n,\eps) = n$. Therefore after applying the recursive step five times, we obtain a test with complexity $\tilde{\Oh}(\eps^{(2 \delta)^{1/\log 3/4 - 2} + 2} n^{2/5+\nicefrac{4}{5} \cdot \delta})$.  

\paragraph{Correctness.}
We now show that the \textsc{$D_m$-consistency$(n,5)$} test is an $\eps$-tester for \textsc{$D_m$-consistency$(n)$} with bounded error $1/6$. By the definition, we need to show that if $T$ is $D_m$-consistent, the test will accept it with probability $> 1-1/6$, and if $T$ is accepted with probability $> 1/6$, $T$ is $\mathbf{C} \cdot \eps$-close to $D_m$-consistent for some constant $\mathbf{C} > 0$.

We start with the first part of the claim. Suppose that if $T$ is $D_m$-consistent, then it is accepted by the \textsc{$D_m$-consistency$(n,k-1)$} test with probability $> 1-\alpha$. By the union bound and Lemma~\ref{lm:matching_T_in_D2} it then follows that \textsc{$D_m$-consistency$(n,k)$} test will accept $T$ with probability $> 1-\alpha^2 \cdot\eps^{-1}\log n - 1/n$ (the test can err either in the inter-block matching, or in one of the $\eps^{-1} \log n$ calls to \textsc{$D_m$-consistency$(n,k-1)$}). Since for $k=0$ the error probability $\alpha = 0$, we obtain that the error probability for $k=5$ is less than $1/6$.

We now show the second part of the claim by induction on $k$. Suppose that the following is true: If the \textsc{$D_m$-consistency$(n,k-1)$} test accepts $T$ with probability $> 1/6$, then it is $\mathbf{C}_{k-1} \cdot \eps$-close to $D_m$-consistent for some constant $\mathbf{C}_{k-1} > 0$. We will now show that if the \textsc{$D_m$-consistency$(n,k)$} test accepts $T$ with probability $> 1/6$, then it is $\mathbf{C}_{k} \cdot \eps$-close to $D_m$-consistent for some constant $\mathbf{C}_k > 0$. This will conclude the proof of Theorem~\ref{thm:D_m_consistency}. We do this in three steps as described below: First, we show how to make all the blocks $D_m$-consistent, secondly, we show that we can match almost all locally excess parentheses in the resulting string, and finally, we show how to modify the remaining locally excess parentheses so that we can match them as well.

\paragraph*{Step 1 - Making all blocks $D_m$-consistent.} Since the \textsc{$D_m$-consistency$(n,k)$} test accepts $T$ with probability $> 1/6$, at most $\eps$-fraction of the blocks can be accepted by the  \textsc{$D_m$-consistency$(n,k-1)$} test with probability $\le 1/6$. Indeed, if there were more than $\eps$-fraction of such blocks, one of them would be selected with probability $> 1-1/n$. Recall that we run the \textsc{$D_m$-consistency$(n,k-1)$} test on this block twice (Algorithm~\ref{alg:C_m}, Step 4(b)) . It means that it will be accepted by the test with probability $\le 1/36$. By the union bound we obtain that in this case $T$ would be accepted with probability $< 1/6$, a contradiction. It follows that at least a $(1-\eps)$-fraction of the blocks would be accepted by the \textsc{$D_m$-consistency$(n,k-1)$} test with probability $> 1/6$, and by our assumption they are $\mathbf{C}_{k-1} \cdot \eps$-close to $D_m$-consistent.

We now explain how we modify the blocks to make them $D_m$-consistent. First we consider all blocks that are $\mathbf{C}_{k-1} \cdot \eps$-far from $D_m$-consistent. For every such block $B$, there is a unique perfect matching on the non-excess parentheses. If a pair of matched non-excess parentheses have different types, we modify one of them accordingly. Note that this procedure does not change the set of excess parentheses of such blocks. In total, we modify at most $\eps n / 2$ parentheses (up to $b/2$ in each such block). We now consider the blocks that are $\mathbf{C}_{k-1} \cdot \eps$-close to $D_m$-consistent. By definition we can make each such block $D_m$-consistent by modifying $\le \mathbf{C}_{k-1} \cdot \eps b$ parentheses in it. In total, we modify at most $\mathbf{C}_{k-1} \cdot \eps n$ parentheses. We denote the resulting string by $T'$. From Equation~\ref{eq:excess} we immediately obtain the following observation, that will be important for further analysis.

\begin{observation}
After modifying $\le \mathbf{C}_{k-1} \cdot \eps b$ characters of a block, the set of excess opening/closing parentheses in it can change by at most $2 \mathbf{C}_{k-1} \cdot \eps b$ parentheses. 
\end{observation}

We finally obtain that the sets of excess parentheses of $T$ and $T'$ differ by at most $2 \mathbf{C}_{k-1} \cdot \eps n$ parentheses.

\paragraph*{Step 2 - Partial matching of locally excess parentheses.} We now build a matching on the locally excess parentheses of $T'$. We first consider the initial non-modified string $T$. The inter-block matching test must accept $T$ with probability $> 1/6$ and therefore by Lemma~\ref{lm:matching_T_in_D2} there is a matching on locally excess parentheses of $T$ such that: (a) Any two matched parentheses have the same type and the distance between them is even; (b) There are at most $\Oh(\eps n)$ unmatched locally excess parentheses. We now consider an induced matching on excess parentheses of $T'$. Namely, we match two locally excess parentheses of $T'$ if they are matched in $T$ and if they were not modified during the first step. From Lemma~\ref{lm:global_tests_accepts} it follows that $T'$ will contain at most $\Oh(\eps n)$ non-matched locally excess parentheses.

\paragraph*{Step 3 - Modifying non-matched excess parentheses.} The string $T'$ is now composed of four types of consecutive substrings: (a) Substrings that belong to $D_m$; (b) Locally excess parentheses in one block that are matched with locally excess parentheses in another block; (c) At most $\Oh(\eps n)$ locally excess parentheses that are not matched (see Step 2); and (d) At most $\Oh(\eps n)$ excess parentheses of $T'$ that are not excess parentheses of $T$ (see Step 1); (e) Excess parentheses of $T$. 

Let $T''$ be the string obtained from $T'$ by removing all substrings of type (a). Note that by removing such substrings we do not change the parity of the distance between any two matched excess parentheses. We show how to modify $T''$ in a recursive way. Let $t', t''$ be two matched substrings of excess parentheses such that between them there is only one substring $\tau$ of parentheses of types (c) or (d). It follows that $S'' = p' t' \tau t'' p''$ for some strings $p'$ and $p''$. The length of $\tau$ is even. We replace it with an arbitrary string in $D_m$, and then remove $t' \tau t''$ from $S''$ and continue recursively. In the end we will obtain a set of excess parentheses of $T$. This concludes the proof of Theorem~\ref{thm:D_m_consistency}.

\section{Algorithm for \textsc{Substring $\eps$-matching}}
\label{sec:substring_matching}
In this section we show an algorithm for $\textsc{Substring $\eps$-matching}$ with bounded error $1/3$ and query complexity $\tilde{\Oh}(\eps^{2\delta^{1/{\log 3/4 - 2}}+4} x^{1/2+\delta})$. Theorem~\ref{th:substring_matching} will immediately follow, as we can repeat the algorithm a logarithmic number of times to boost the probability. 

\subsection{Algorithm for $\textsc{Substring $\eps$-matching}$ with bounded error $1/3$}
Recall that the algorithm receives as an input two substrings $S_1, S_2$ of $T$ of length at most $x$, and must accept $S_1, S_2$ if they $\eps$-match, and reject if at most $e_0 (S_1) - 30 \eps x$ excess opening parentheses of $S_1$ can be sequentially matched in $S_2$. The algorithm consists of three recursive procedures: Procedures \textsc{QueryLeft$(x, \eps, k)$} and \textsc{QueryRight$(x, \eps, k)$} query a subset of characters of strings $S_1$ and $S_2$, and the third procedure, \textsc{MakeDecision$(x, \eps, k)$} accepts or rejects $(S_1, S_2)$ using the queried characters. We give the pseudocode of our solution in Algorithm~\ref{alg:substring_matching}.

\begin{algorithm}
\caption{\textsc{Substring $\eps$-matching}}
\label{alg:substring_matching}
Input: Two substrings $S_1$, $S_2$ of a string $T$ of length $\le x$
\begin{enumerate}
\item $k := \lceil \log_{3/4} 2\delta\rceil$
\item Run \textsc{QueryLeft$(x, \eps, k)$} for $S_1$
\item Run \textsc{QueryRight$(x, \eps, k)$} for $S_2$
\item \textsc{MakeDecision$(x, \eps, k)$}
\end{enumerate}
\end{algorithm} 

We now describe each procedure in turn.

\paragraph{Procedures \textsc{QueryLeft$(x, \eps, k)$} and \textsc{QueryRight$(x, \eps, k)$}.}
Let $x' = x^{3/4}$ and $\eps' = \eps / 30$. 
Procedure \textsc{QueryLeft$(x, \eps, k)$} starts by running the excess numbers preprocessing on $S_1$ for precision $(\eps')^2 x'$. Next, it partitions $S_1$ into non-overlapping blocks starting from the right. If the approximate number of excess parentheses in $S_1$ is less than $10 \eps' x'$, the partitioning of $S_1$ is defined to contain a single block equal to $S_1$ itself. Otherwise, it must satisfy the following two properties: (1) The approximate number of excess opening parentheses in the $m$ rightmost blocks of $S_1$ is between $(10  m - \eps') \cdot \eps' x'$ and $(10  m + \eps') \cdot \eps' x'$; and (2) The approximate number of excess closing parentheses in these blocks is at most $(\eps')^2 x'$. (Note that such a partitioning always exists because, for example, we can choose the $m$-th block to be the substring bounded by the $10(m-1)$-th and $10 m$-th excess parentheses. The procedure might choose another partitioning, but this shows that it will have at least one possible choice.) We call blocks of length $\le x'$ \emph{dense}.  Finally, the procedure deletes the leftmost and the rightmost blocks of $S_1$. 

\begin{procedure}
\caption{\textsc{QueryLeft$(x, \eps, k)$}}
\label{alg:substring_matching_QL}
Input: Substring $S_1$ of $T$ of length $\le x$\\
Output: Partitioning of $S_1$, a sequence $\mathcal{L} = \{\mathcal{L}_t\}$ of subsets of dense blocks, queried characters
\begin{enumerate}
\item If $k = 0$, query all characters of $S_1$
\item $x' := x^{3/4}$, $\eps' := \eps / 30$
\item Run excess numbers preprocessing for $S_1$ with precision $(\eps')^2 x'$
\item Partition $S_1$ into blocks containing approximately $10 \eps' x'$ excess opening parentheses, and then delete the leftmost and the rightmost blocks
\item Select a sequence $\mathcal{L} = \{\mathcal{L}_t\}$ of $\tilde{\Theta} (\eps^{-1})$ random subsets  of dense blocks of $S_1$ of size $\tilde{\Theta}((\eps')^{-1} \sqrt{x / x'})$
\item Run \textsc{QueryLeft$(x', 2 (\eps')^2, k-1)$} for each selected block $\Theta (\log (x \cdot \eps^{-1} \log x) )$ times 
\end{enumerate}
\end{procedure} 

Let $\mathbf{A}$, $\mathbf{B}$, and $\mathbf{C}$ be positive constants to be defined later. The procedure selects $\mathbf{A} \cdot \eps^{-1} \log x$ random subsets $\mathcal{L}_t$ of dense blocks of $S_1$, where each subset has size $\mathbf{B} \cdot (\eps')^{-1} \sqrt{x / x'} \log (x \cdot \eps^{-1} \log x)$. For each selected dense block $B$ the procedure runs \textsc{QueryLeft$(x', 2 (\eps')^2, k-1)$} independently $\mathbf{C} \cdot \log (x \cdot \eps^{-1} \log x)$ times over $B$. The pseudocode is given in Procedure~\ref{alg:substring_matching_QL}.

Similar to above, procedure \textsc{QueryRight$(x, \eps, k)$} starts by running the excess numbers preprocessing on $S_2$ for precision $(\eps')^2 x'$. For technical reasons that will become clear later, we consider not just one partitioning of $S_2$, but a number of them. Namely, we consider a separate partitioning for each shift $\tau = (\eps')^2 x', 2 (\eps')^2 x',\ldots, 12\cdot \eps x'$ (in total, we have $12 \eps / (\eps')^2 = 12 \cdot 30 / \eps'$ shifts). For each $m \ge 0$ the approximate number of excess closing parentheses in the $(m+1)$ leftmost blocks of the partitioning must be between $\tau + (10 m - \eps') \cdot \eps' x'$ and $\tau + (10m + \eps') \cdot \eps' x'$, and the approximate number of excess opening parentheses must be at most $(\eps')^2 x'$. 

\begin{procedure}
\caption{\textsc{QueryRight$(x, \eps, k)$}}
\label{alg:substring_matching_QR}
Input: Substring $S_2$ of $T$ of length $\le x$\\
Output: Partitionings of $S_2$, sequences $\mathcal{R}(\tau)  = \{\mathcal{R}_t (\tau)\}$ of subsets of dense blocks, queried characters
\begin{enumerate}
\item If $k = 0$, query all characters of $S_2$
\item $x' := x^{3/4}$, $\eps' := \eps / 30$
\item Run excess numbers preprocessing for $S_2$ with precision $(\eps')^2 x'$
\item For each shift $\tau \in \{(\eps')^2 x, 2(\eps')^2 x, \ldots, 12 \eps x\}$:
\begin{enumerate}
\item Partition $S_2$ into blocks containing approximately $10 \eps' x'$ excess closing parentheses, except for the first block containing approximately $\tau$ excess closing parentheses
\item Select a sequence $\mathcal{R}(\tau)  = \{\mathcal{R}_t (\tau)\}$ of $\tilde{\Theta} (\eps^{-1})$ sets  of dense blocks of $S_2$ of size $\tilde{\Theta}((\eps')^{-1} \sqrt{x / x'})$
\item Run \textsc{QueryRight$(x', 2 (\eps')^2, k-1)$} for each selected block $\Theta (\log (x \cdot \eps^{-1} \log x) )$ times 
\end{enumerate}
\end{enumerate}
\end{procedure} 

For the partitioning of $S_2$ corresponding to a shift $\tau$, the procedure selects $\mathbf{A} \cdot \eps^{-1} \log x$ random subsets $\mathcal{R}_t (\tau)$ of dense blocks of $S_2$ of size $\mathbf{B} \cdot (\eps')^{-1} \sqrt{x / x'} \log (x \cdot \eps^{-1} \log x)$ each. For each selected dense block $B$ the procedure runs \textsc{QueryRight$(x', 2 (\eps')^2, k-1)$} independently $\mathbf{C} \cdot \log (x \cdot \eps^{-1} \log x)$ times. The pseudocode is given in Procedure~\ref{alg:substring_matching_QR}.

\paragraph{Procedure \textsc{MakeDecision$(\eps, x, k)$.}}
We finally explain how we use the queried indexes to test $S_1$ and $S_2$. If $k = 0$, \textsc{QueryLeft$(\eps, x, k)$} and \textsc{QueryRight$(\eps, x, k)$} know all characters of $S_1$ and $S_2$ and we can use a naive deterministic algorithm to decide whether $S_1$ and $S_2$ $\eps$-match. If $\hat{e}_0(S_1) < 10 \eps x$, we always accept $S_1$ and $S_2$.

Suppose now that $k > 0$ and $\hat{e}_0(S_1) \ge 10 \eps x$.
For each partitioning of $S_2$ we consider all possible substrings $X$ that start at some block border. We process each of them in turn and accept $(S_1, S_2)$ if at least one of the substrings $X$ is accepted. If the difference between the approximate number of excess opening parentheses in $S_1$, $\hat{e}_0(S_1)$, and excess closing parentheses in $X$, $\hat{e}_1(X)$ is larger than $4 (\eps')^2 x'$, $X$ is rejected, and otherwise we continue to the next step. $X$ is tested in $\mathbf{A} \cdot \eps^{-1} \log x$ iterations, and we accept $X$ if and only if it is accepted at each iteration. Suppose that $X$ starts at a block border of a partitioning for a shift~$\tau$. We enumerate the blocks in $S_1$ from right to left and the blocks in $X$ from left to right. At iteration $t = 1, 2, \ldots, \mathbf{A} \cdot \eps^{-1} \log x$ we find a rank $m$ such that the $m$-th block $B_1$ of $S_1$ is in the subset $\mathcal{L}_t$, and the $m$-th block $B_2$ of $S_2$ is in $\mathcal{R}_t(\tau)$. Finally, we run the \textsc{MakeDecision$(x', 2(\eps')^2, k-1)$} procedure on $(B_1, B_2)$ independently $\textbf{C} \cdot \log (x \cdot \eps^{-1})$ times; and if the blocks are rejected for the majority of iterations, reject $X$. (See Procedure~\ref{alg:substring_matching_MD}.)

\begin{procedure}
\caption{\textsc{MakeDecision$(x, \eps, k)$}}
\label{alg:substring_matching_MD}
Input: Substrings $S_1$, $S_2$ of a string $T$ of length $\le x$; outputs of \textsc{QueryLeft$(x, \eps, k)$} run on $S_1$ and of \textsc{QueryRight$(x, \eps, k)$} run on $S_2$ 
\begin{enumerate}
\item If $k = 0$, use the trivial algorithm to decide whether $S_1$ and $S_2$ $\eps$-match
\item $x' := x^{3/4}$, $\eps' := \eps / 30$
\item If $\hat{e}_0(S_1) < 10 \eps x$, accept $S_1, S_2$
\item For each partition of $S_2$ and for each substring $X$ starting at the partition's block border:
\begin{enumerate}
\item If $|\hat{e}_0(S_1) -\hat{e}_1(X)| > 4 (\eps')^2 x'$, reject $X$
\item Find $\Theta(\eps^{-1} \log x)$ pairs of dense blocks of $S_1$ and $X$ that have equal ranks using sets $\mathcal{L}_t$ and $\mathcal{R}_t (\tau)$. For each such pair $(B_1, B_2)$:
\begin{enumerate}
\item Run \textsc{MakeDecision$(x', 2(\eps')^2, k-1)$} on $(B_1, B_2)$ $\Theta (\log (\eps^{-1} x \log x) )$ times
\item If $(B_1, B_2)$ is rejected for the majority of iterations, reject $X$
\end{enumerate}
\item Accept $(S_1, S_2)$ if $X$ is not rejected
\end{enumerate}
\end{enumerate}
\end{procedure}

\subsection{Analysis (proof of Theorem~\ref{th:substring_matching})} 
We now show complexity and correctness of the algorithm. 

\begin{lemma}\label{lm:substring_matching_complexity}
The query complexity of Algorithm~\ref{alg:substring_matching} is $\tilde{\Oh}(\eps^{(2 \delta)^{1/\log 3/4 - 2} + 4} x^{1/2+\delta})$.
\end{lemma}
\begin{proof}
Since the algorithm queries $S_1$ and $S_2$ only during the procedures \textsc{QueryLeft$(x, \eps, k)$} and \textsc{QueryRight$(x, \eps, k)$}, it suffices to estimate their query complexity. Let us first analyse one recursive step. The description of the procedures implies that if the query complexities of \textsc{QueryLeft$(x, \eps, k-1)$} and \textsc{QueryRight$(x, \eps, k-1)$} are bounded by $f (x, \eps)$, then the sum of query complexities of the procedures is $\tilde{\Oh}(\eps^{-4} \sqrt{x} + f(x^{\nicefrac{3}{4}}, 2 (\eps / 30)^2) \cdot \eps^{-3} x^{\nicefrac{1}{8}} \log^2 \eps^{-1} x)$. Therefore, if $f(x, \eps) = \tilde{\Oh}( \eps^{y} x^{\nicefrac{1}{2} + z})$, then the sum of query complexities of \textsc{QueryLeft$(x, \eps, k)$} and \textsc{QueryRight$(x, \eps, k)$} is $\tilde{\Oh}(\eps^{2y-4} x^{\nicefrac{1}{2} + \nicefrac{3}{4} \cdot z})$ (we use $\log^2 \eps^{-1} < \eps^{-1}$ and omit all $\log x < \log n$ terms). After $r$ iterations the degree of $x$ becomes $1/2+ (3/4)^r \cdot y$, and the degree of $\eps$ becomes $2^{r} (z - 4) + 4$.

Recall that the query complexity of the trivial algorithm for $k = 0$ is $f(x, \eps) = x = \eps^{y} x^{\nicefrac{1}{2} + z}$, where $y = 0$ and $z = 1/2$. Therefore, after $r = \log_{3/4} 2\delta$ recursive steps we obtain an algorithm with query complexity $\tilde{\Oh}(\eps^{(2 \delta)^{1/\log 3/4 - 2} + 4} x^{1/2+\delta})$.   
\end{proof}

We now show that Algorithm~\ref{alg:substring_matching} has bounded error $1/3$. 

\begin{lemma}\label{lm:matching_T_in_D2}
If $S_1, S_2$ $\eps$-match, Algorithm~\ref{alg:substring_matching} accepts them with probability $>2/3$.
\end{lemma}
\begin{proof}
We will show that if \textsc{MakeDecision$(\eps, x, k-1)$} accepts $\eps$-matching strings with probability $> 2/3$, then \textsc{MakeDecision$(\eps, x, k)$} accepts $\eps$-matching strings with probability $> 2/3$ as well. The claim will then follow by induction, as \textsc{MakeDecision$(\eps, x, 0)$} (the trivial algorithm) accepts $\eps$-matching strings with probability $1$.

By the definition of $\eps$-matching, all but at most $7 \eps x$ leftmost and $5 \eps x$ rightmost excess opening parentheses of $S_1$ can be sequentially matched in $S_2$. Assume that all excess numbers are approximated correctly, which is true with probability $> 1 - 1/2n$ and consider a subsequence $\pi$ of excess opening parentheses in $S_1$ that contains all excess opening parentheses of $S_1$ except for those that belong to its leftmost and rightmost blocks. We rank the parentheses in $\pi$ from right to left. Let $p$ be the leftmost excess parenthesis in $S_2$ matched with a parenthesis of $\pi$. We rank the excess closing parentheses of $S_2$ from left to right, starting from $p$. Across all partitions of $S_2$, take the rightmost block border preceding $p$ and let $X$ be a substring of $S_2$ starting at it. We will show that $X$ is accepted by the algorithm with probability $> 2/3$, from which the claim follows. Note that the number of excess parentheses of $S_2$ between the left-endpoint of $X$ and $p$ is at most $3 (\eps')^2 x'$. 

Recall that \textsc{MakeDecision$(x, \eps, k)$} finds a subset of dense (that is, of length $\le x'$) blocks of $S_1$ and $X$ that have equal ranks and runs \textsc{MakeDecision$(x', 2 (\eps')^2, k-1)$} on them independently $\textbf{C} \cdot \log (x \cdot \eps^{-1}) $ times. Consider the $m$-th block $B_1$ of $S_1$ (with the rightmost block deleted) and the $m$-th block $B_2$ of $X$. We will show that all but at most $14 (\eps')^2 x'$ leftmost and $10(\eps')^2 x'$ rightmost excess parentheses of $B_1$ can be sequentially matched in $B_2$, which means that $B_1, B_2$ $2 (\eps')^2$-match. The rank of the rightmost excess opening parenthesis in $B_1$ is at least $(10 m - 2 \eps') \eps' x'$. The rank of the leftmost excess opening parenthesis in $B_1$ is at most $(10 (m+1) + 2 \eps') \eps' x'$. Also, $B_1$ can end with at most $2 (\eps')^2 x'$ excess opening parentheses that are not excess parentheses of $S_1$. The rank of the leftmost excess closing parenthesis in $B_2$ is at most $(10 m + 2 \eps') \eps' x'$. The rank of the rightmost excess closing parenthesis in $B_2$ is at least $(10 (m+1) - 5 \eps') \eps' x'$. Consequently, all but at most $7 (\eps')^2 x' < 14 (\eps')^2 x'$ leftmost excess parentheses and $6 (\eps')^2 x' < 10  (\eps')^2 x$ rightmost excess parentheses of $x$ can be matched in $B_2$ as required. By our assumption, \textsc{MakeDecision$(x', 2 (\eps')^2, k-1)$} accepts $(B_1, B_2)$ with probability $>2/3$, and therefore we can select the constant $\mathbf{C}$ so that $B_1$ and $B_2$ are accepted with probability at least $1 - 1 / (6 \mathbf{A} \eps^{-1} \log x)$. From the union bound it follows that all of the $\mathbf{A} \cdot \eps^{-1} \log x$ pairs of blocks for which we run the recursive call will be accepted with probability at least $1-1/6$, and consequently $S_1$ and $X$ will be accepted with probability $> 2/3$.  
\end{proof}

We now show by contrapositive that if at most $e_0(S_1)-30 \eps x$ excess opening parentheses of $S_1$ can be matched sequentially in $S_2$, $S_1$ and $S_2$ will be rejected with probability $> 2/3$. We start with an auxiliary lemma. Recall that at each iteration $t$ the procedure \textsc{MakeDecision$(x, \eps, k)$} considers a block partitioning of $S_2$ with some shift $\tau$ and chooses a subset $\mathcal{L}_t$ of blocks of $S_1$ and a subset $\mathcal{R}_t(\tau)$ of blocks of the partition of $S_2$. Using these two subsets, it tests $S_1$ and each substring $X$ of $S_2$ that starts at a block border of the partition of $S_2$. We rank the blocks in $X$ from left to right. Blocks of $S_1$ are ranked in the reverse order, from right to left. Intuitively, two blocks of $S_1$ and $X$ have equal ranks if they contain many parentheses that must be matched, and therefore we can recurse on them. Below we show that $\mathcal{L}_t$ and $\mathcal{R}_t(\tau)$ contain such blocks with high probability. 

\begin{lemma}\label{lm:birthday}
Assume that the excess numbers preprocessing for $S_1$ and $S_2$ did not error and that $\hat{e}_0(S_1) \ge 10 \eps x$. For all $t$ and for all substrings $X$ of $S_2$ that are not rejected at Step 4(b) of \textsc{MakeDecision$(x, \eps, k)$}, the subsets $\mathcal{L}_t$ and $\mathcal{R}_t(\tau)$ contain a pair of dense blocks with equal ranks with probability $> 1 - 1/9 x^2$.
\end{lemma}
\begin{proof}
Let $\mathcal{D}$ be the set of ranks $m$ of blocks such that both the $m$-th block of $S_1$ and the $m$-th block of $X$ are dense. From the assumption of the lemma we have $e_0(S_1) \ge 9 \eps x$ and since $X$ is not rejected at Step 4(b), $e_1(X) > 9 \eps x$ as well. Therefore, the total number of all blocks in $S_1$ or $X$ is at least $9 \eps x / 11 \eps' x'$. The total number of non-dense blocks in both strings is at most $2 x /x' = 22 \eps' x / 11 \eps' x' < \eps x / 11 \eps' x'$. Therefore, $|\mathcal{D}| \ge 8 \eps x / 11 \eps' x'$. On the other hand, the total number of blocks (and, in particular, dense blocks) in $S_1$ and $X$ is at most $x / 8 \eps' x'$. It follows that $|\mathcal{D}| \le x / 8 \eps' x'$.

Recall that both $\mathcal{L}_t$ and $\mathcal{R}_t(\tau)$ have size $\mathbf{B} \cdot (\eps')^{-1} \sqrt{x/ x'} \log (\eps^{-1} x \log x)$. We view the sets as $\mathbf{B} \cdot \log (\eps^{-1} x \log x)$ experiments during which we select two subsets of $\mathcal{D}$. Note that each dense block is selected with probability at least $((\eps')^{-1}  \sqrt{x/x'}) / (x / 8 \eps' x') \ge 8 / \sqrt{x / x'}$. Therefore, the expectation of the size of the selected subsets of $\mathcal{D}$ is at least $8 |\mathcal{D}| / \sqrt{x / x'}$. From the lower bound on $|\mathcal{D}|$ it follows that the latter is at least $2 \sqrt{|\mathcal{D}|}$, and therefore the size of the selected subsets is at least $\sqrt{|\mathcal{D}|}$ with probability $> 3/4$ (this is a rough bound which is sufficient for our purposes). Recall that the Birthday paradox claims that any two subsets of $\mathcal{D}$ of size $\sqrt{|\mathcal{D}|}$ sampled uniformly without replacement contain equal elements with probability $> 1/2$. Therefore, in each experiment there is a pair of dense blocks with equal ranks with probability $> 1/4$. Since we repeat the experiments $\mathbf{B} \cdot \log (\eps^{-1} x \log x)$ times, the probability that at least one of them is successful is at least $1 - 1 / (9 \mathbf{A} \cdot \eps^{-1} x^2 \log x)$ for a sufficiently large constant $\mathbf{B}$. By the union bound over all $t$ the lemma holds with probability $> 1 - 1/9x^2$. 
\end{proof}

\begin{lemma}\label{lm:matching}
If less than $e_0 (S_1) - 30 \eps x$ excess opening parentheses of $S_1$ can be matched sequentially in $S_2$, then Algorithm~\ref{alg:substring_matching} rejects $S_1$ and $S_2$ with probability $> 2/3$.
\end{lemma}
\begin{proof}
We show the claim by induction on $k$. Suppose that $\textsc{MakeDecision$(\eps, x, k-1)$}$ rejects with probability $> 2/3$ if run on two strings $B_1, B_2$ such that less than $e_0 (B_1) - 30 \eps x$ excess opening parentheses of $B_1$ can be matched sequentially in $B_2$. We will show it implies that $\textsc{MakeDecision$(x, \eps, k)$}$ will reject $S_1, S_2$ with probability $> 2/3$. The lemma then follows, as the base case ($k = 0$) obviously holds.   

Consider some substring $X$ of $S_2$. We call a block of $S_1$ of rank $m$ \emph{bad} (with respect to $X$) if we cannot match more than $60 (\eps')^2 x'$ excess opening parentheses of it in the $m$-th block of $X$. Suppose that at most $\eps /3$-fraction of blocks of $S_1$ are bad. We show that in this case we can match more than $e_0 (S_1) - 30 \eps x$ excess opening parentheses of $S_1$ in $X$. Indeed, the total number of non-dense blocks of $S_1$ is at most $x / x'$ and hence the total number of excess parentheses in non-dense blocks is at most $12 \eps' x = 12 \eps x / 30 < \eps x / 2$. Consider the set of dense blocks of $S_1$. The leftmost and the rightmost blocks of $S_1$ can be dense but contain at most $24 \eps x$ excess opening parentheses. By our assumption, among the remaining blocks there is at most $\eps/3$-fraction of bad blocks. Therefore, the number of excess parentheses in such blocks is at most $(x/8\eps' x') \cdot (\eps / 3) \cdot 12 \eps' x' \le \eps x / 2$. On the other hand, we can match all but $60 (\eps')^2 x'$ excess opening parentheses in each of the remaining blocks, or at most $60 (\eps')^2 x / 8 = 5 \eps x / 2$ parentheses in total. Therefore, the total number of unmatched excess parentheses is less than $30 \eps x$ as claimed. 

It therefore suffices to show that if $S_1$ contains more than $\eps /3$-fraction of bad blocks for each substring $X$ of $S_2$, then it will be rejected with probability $> 2/3$. Equivalently, we can show that the probability to accept $S_1$ and $S_2$ is at most $1/3$. We can accept the strings either because we made an error while approximating the excess numbers (which can happen with probability $< 1/9$) or because $S_1$ and some substring $X$ of $S_2$ are erroneously accepted. Since the length of $S_2$ is at most $x$, it has at most $x^2/2$ substrings. We will show that each of them is accepted with probability $< 2/9x^2$. The claim will follow by the union bound. From Lemma~\ref{lm:birthday} it follows that we will find $\mathbf{A} \cdot \eps^{-1} \log x$ pairs of dense blocks with equal ranks with probability at least $1 - 1/9x^2$. Since at least $\eps /3$-fraction of the blocks of $S_1$ is bad with respect to $X$, we can select the constant $\mathbf{A}$ so that at least one of the bad blocks is selected with probability $> 1 - 1/18x^2$. Finally, we can select the constant $\mathbf{C}$ so that the bad block is rejected by a recursive call to $\textsc{MakeDecision$(\eps, x, k-1)$}$ with probability $> 1 - 1/18x$, which concludes the proof.

Note that we might need different values of $\mathbf{C}$ for this lemma and Lemma~\ref{lm:matching_T_in_D2}. We take the maximum of the values to ensure both lemmas.
\end{proof}

\section{Lower bounds for \textsc{Truestring equivalence} and \textsc{$D_m$-membership}}\label{sec:lb}
In this section we prove the following lower bound for testing truestring equivalence.

\begin{theorem}\label{th:TE_lower}
Testing truestring equivalence requires at least $\Omega(n^{1/5})$ queries.
\end{theorem}

Since the \textsc{Truestring equivalence$(n)$} problem can be reduced to the \textsc{$D_m$-membership$(4n)$} problem by Lemma~\ref{lm:reduction}, we immediately obtain a similar lower bound for testing $D_m$-membership.

\begin{corollary}\label{cor:D_m_lower}
Testing $D_m$-membership requires at least $\Omega(n^{1/5})$ queries.
\end{corollary}

For the lower bound construction let us introduce several definitions. Recall that for a string $w\in\{0,1,\nochar\}^*$, its {\em truestring} $T(w)$ is the subsequence resulting from deleting all ``$\nochar$'' characters. Given a string $u\in\{0,1\}^n$ and a set $U\in\binom{\{1,\ldots,2n\}}{n}$, we denote by $S(u,U)$ the unique string $w\in\{0,1,\nochar\}^{2n}$ for which $U=\{i:w[i]=\mbox{``}\nochar\mbox{''}\}$ and $u=T(w)$.

\begin{definition}[Positive and negative distributions]
We let $u\in\{0,1\}^n$ be chosen uniformly at random (every $u[i]$ independently), and let $u'\in\{0,1\}^n$ be the string resulting from replacing $u[i]$ with another uniformly and independently random member of $\{0,1\}$ for every $2n/5<i<3n/5$. Let $U\in\binom{\{1,\ldots,2n\}}{n}$ be a random set defined by choosing independently and uniformly whether $i\in U$ and $2n+1-i\not\in U$, or $i\not\in U$ and $2n+1-i\in U$, for every $1\leq i\leq n$. Let $U'\in\binom{\{1,\ldots,2n\}}{n}$ be a second set chosen independently using the same distribution as that used for the choice of $U$. For the distribution $\mathcal{D}_P$, we set $w=S(u,U)$ and $w'=S(u,U')$. For the distribution $\mathcal{D}_N$, we set $w=S(u,U)$ and $w'=S(u',U')$.
\end{definition}

\begin{lemma}\label{lem:far}
$\mathcal{D}_P$ is supported over string pairs that are truestring equivalent, while $\mathcal{D}_N$ with probability $1-o(1)$ produces a pair that is $1/200$-far from truestring equivalence.
\end{lemma}
\begin{proof}
The first part of the statement is immediate. The second part follows from the fact that two strings of length $n/5$ drawn uniformly and independently at random will have an edit distance between them of at least $n/100$. For showing this consider all $\binom{n/5}{n/100}^2=o(2^{16n/100})$ possible ways of deleting $n/100$ characters from the first string and $n/100$ characters from the second string. For every such possibility the probability for the remaining strings to match is $\Theta(2^{-19n/100})$. A union bound concludes the argument.
\end{proof}

In the rest of the section we prove the next lemma, which by Yao's argument implies (together with Lemma \ref{lem:far}) Theorem~\ref{th:TE_lower}.
\begin{lemma}\label{lem:undistinguishable}
Any deterministic algorithm making $o(n^{1/5})$ queries will have acceptance probabilities for $\mathcal{D}_P$ and $\mathcal{D}_N$ that differ by $o(1)$.
\end{lemma}

Toward the proof of this lemma, we use the following definitions in the analysis.

\begin{definition}[True index]
	Given a string $w$ and $1\leq i\leq 2n$ for which $w[i]\neq\nochar$, the {\em true index} $t_w(i)$ is defined as $|\{j\leq i:w[j]\neq\nochar\}|$. In other words, it is the index $j$ such that $w[i]$ determines $(T(w))[j]$.
\end{definition}

The next definition is adapted for our particular distributions $\mathcal{D}_P$ and $\mathcal{D}_N$, where for every $i$ and $i'$ exactly one of $w[i]$ and $w[2n+1-i]$ is $\nochar$ and exactly one of $w'[i']$ and $w'[2n+1-i']$ is $\nochar$.

\begin{definition}[Matching indexes]
	Given the strings $w$ and $w'$, and the indexes $1\leq i\leq 2n$ and $1\leq i'\leq 2n$, we call $(i,i')$ {\em matching} with respect to $(w,w')$ if either $w[i],w'[i']\neq\nochar$ and $t_w(i)=t_{w'}(i')$, or $w[2n+1-i],w'[2n+1-i']\neq\nochar$ and $t_w(2n+1-i)=t_{w'}(2n+1-i')$  (we omit below the pair notation, e.g.\ ``$(i,i')$'' and ``$(w,w')$'', whenever it is clear which index among $i$ and $i'$ refers to $w$ and which index refers to $w'$).
\end{definition}

Ironically, the analysis of a deterministic adaptive algorithm becomes much easier if we allow the input query oracle to ``leak'' additional information when answering a query. We define what we mean below. The following definition is valid for any finite domain, but for our purposes we define a domain $D := \{1,\ldots,2n\}\times\{1,2\}$, referring to two length-$2n$ strings. 

\begin{definition}[Super-oracle]
	A \emph{super-oracle} for an input $W$ over the domain $D$ is an algorithm, intended to be called whenever a query to $W$ is made, satisfying the following features.
	\begin{itemize}
		\item The algorithm holds $W$ and possibly other variables (i.e.\ it is not stateless), which are given initial values before the first query.
		\item Whenever a query $i\in D$ is made, the algorithm provides a set of indexes $I$ and all values of $W|_I$. It is mandated that $i\in I$, and assumed (without loss of generality) that $I$ also contains all indexes provided in the previous query (except when responding to the very first query).
	\end{itemize}
\end{definition}

It is immediate that a testing algorithm (working with a ``regular'' oracle) can be converted to a testing algorithm working with a super-oracle $\mathcal{Q}$, by simply ignoring the additional information, which leads to the following observation.

\begin{observation}\label{obs:supenough}
	To prove Lemma \ref{lem:undistinguishable}, it is enough to construct a super-oracle $\mathcal{Q}$, and prove that the distributions over the transcript of any algorithm $\mathcal{A}$ with $q=o(n^{1/5})$ queries running against~$\mathcal{Q}$, when the input is drawn by either $\mathcal{D}_P$ or $\mathcal{D}_N$, are $o(1)$-close to each other.
\end{observation}

Another assumption that is without loss of generality, and that we make from now on, is that the algorithm $\mathcal{A}$ never makes a query which has already been revealed in response to a previous query (the algorithm can just read internally the value that has already been given), and in particular we assume that the algorithm terminates immediately (and gives the correct answer) when $\mathcal Q$ reveals the entirety of the input $W$. We denote this ``total revelation'' event by ``$\bot$''.

The following observation about deterministic algorithms interacting with queries is well-known for direct queries, and also holds (with the same proof) for algorithms interacting with a specific super-oracle.

\begin{observation}\label{obs:lastrev}
	For a deterministic algorithm $\mathcal{A}$ working against any super-oracle $\mathcal{Q}$ with any input $W$, under the (without loss of generality) assumptions that the set revealed by $\mathcal{Q}$ in response to a query always contains the previously revealed set and that $\mathcal{A}$ never queries a previously revealed location, the transcript of the algorithm is fully determined by the final revealed set $I$ and the values $W|_I$ of the input over it.
	
	Additionally, the last revealed set $I$ and the values $W|_I$ of the input over it (as provided to $\mathcal{A}$ by~$\mathcal{Q}$) completely determine its state and its transcript so far (including the number of rounds already executed).
\end{observation}

The above observation however is not as useful as it seems at first, since (unlike the setting where an algorithm queries the input without a super-oracle) the probability of an algorithm to reach a set characterized by $I$ and $W|_I$ where $W$ is drawn from some distribution $\mathcal{D}$ is not necessarily the probability of the input to have the prescribed values over $I$. The reason is that $I$ itself might depend on values of the input outside this set. However, we will construct a super-oracle that satisfies an additional feature that helps.

\begin{definition}
	A super-oracle is called \emph{oblivious} if at any point, either the entirety of the input is revealed (the ``$\bot$'' event), or the revealed set $I$ depends only on the history of the algorithm's query locations and the input values on the revealed locations (the conditions for triggering the ``$\bot$'' event can still depend on the entirety of the input).
\end{definition}

This allows us to restore the convenient analysis of algorithms as decision trees, with a caveat concerning the $\bot$ event which we have to analyze separately.

\begin{observation}\label{obs:condprob}
	A deterministic algorithm $\mathcal{A}$ working with an oblivious super-oracle $\mathcal{Q}$ corresponds to a decision tree, where every leaf is labeled with an output (``accept'' or ``reject''), every inner node is labeled by the next revealed set $I$ (since $\mathcal{Q}$ is oblivious, the set is determined only by the history on the path leading to it unless $\bot$ has been triggered), and every edge from a node to a child corresponds to a possible set of values $W|_I$.
	
	Accordingly, when the input is drawn from a distribution $\mathcal{D}$, the probability of reaching a certain leaf whose parent contains the revealed set $I$ (containing the sets of its ancestors) and whose incoming edge corresponds to some value of $W|_I$, is the probability of the intersection of the following two events: The event of an input drawn from $\mathcal D$ to have the corresponding values over $I$, and the negation of the ``$\bot$'' (total reveal) event.
\end{observation}

The exact calculation in the above observation is still complex, but we will only need to have upper bounds as per the following definition.

\begin{definition}
	Given a set $S$ and an additional symbol ``$\bot$'', we say that a distribution $\mathcal{D}$ over $S\cup\{\bot\}$ {\em underlies} a distribution $\mathcal{C}$ over $S\cup\{\bot\}$ if for every $a\in S$ (not including ``$\bot$'') we have $\mathcal{D}(a)\leq\mathcal{C}(a)$.
\end{definition}

The following is immediate.

\begin{observation}\label{obs:underdistance}
	If $\mathcal{D}$ underlies $\mathcal{C}$, both being distributions over $S\cup\{\bot\}$, then the variation distance between them is at most $\mathcal{D}(\bot)$.
\end{observation}
\begin{proof}
	The distance between the two distributions is
	$$\frac12\sum_{a\in S\cup\{\bot\}}|\mathcal{D}(a)-\mathcal{C}(a)|=\sum_{\{a\in S\cup\{\bot\}:\mathcal{D}(a)>\mathcal{C}(a)\}}(\mathcal{D}(a)-\mathcal{C}(a))$$
	In our case, the right hand is a sum over a set that may only contain $\bot$ itself, so the difference is bounded by $\mathcal{D}(\bot)-\mathcal{C}(\bot)\leq \mathcal{D}(\bot)$.
\end{proof}

In our case, we analyse the following distributions.

\begin{definition}
	Given an algorithm $\mathcal{A}$, a super oracle $\mathcal{Q}$ and a distribution $\mathcal{D}$ over inputs, we denote by $\mathcal{A}^{\mathcal{Q}}_{\mathcal{D}}$ the following distribution over $L \cup\{\bot\}$, where $L$ is the set of leaves of the decision tree corresponding to $\mathcal{A}$ and $\mathcal{Q}$ by Observation \ref{obs:condprob}: 
    \begin{itemize}
        \item For $\ell \in L$, $\mathcal{A}^{\mathcal{Q}}_{\mathcal{D}}(\ell)$ equals the probability of reaching $\ell$ (without triggering the $\bot$ event),
        \item $\mathcal{A}^{\mathcal{Q}}_{\mathcal{D}}(\bot)$ equals the probability of the $\bot$ event being triggered (which always implies a correct output by the algorithm).
    \end{itemize}

	For the distributions $\mathcal{D}_P$ and $\mathcal{D}_N$ defined above, we use the shorthand $\mathcal{A}^{\mathcal{Q}}_P$ and $\mathcal{A}^{\mathcal{Q}}_N$ respectively for $\mathcal{A}^{\mathcal{Q}}_{\mathcal{D}_P}$ and $\mathcal{A}^{\mathcal{Q}}_{\mathcal{D}_N}$. 
\end{definition}

We have to be careful about the closeness guarantees if we want to use Observation \ref{obs:supenough}.

\begin{observation}\label{obs:supcareful}
	To prove Lemma \ref{lem:undistinguishable}, it is enough to construct a super-oracle $\mathcal{Q}$, and prove that  $\mathcal{A}^{\mathcal{Q}}_P(\bot)=o(1)$, that $\mathcal{A}^{\mathcal{Q}}_N(\bot)=o(1)$, and  that $\mathcal{A}^{\mathcal{Q}}_P$ and $\mathcal{A}^{\mathcal{Q}}_N$ are $o(1)$-close to each other.
\end{observation}
\begin{proof}
	We use Observation \ref{obs:supenough}, but note that the $\bot$ event does not imply identical outputs of $\mathcal{A}$ under $\mathcal{D}_P$ and $\mathcal{D}_N$. In fact it implies quite the opposite (the $\bot$ event causes $\mathcal{A}$ to provide the correct output, which is different for the two distributions). The distance between the algorithmic behaviors under the two distributions is bounded by $d(\mathcal{A}^{\mathcal{Q}}_P,\mathcal{A}^{\mathcal{Q}}_N)+\frac12\left(\mathcal{A}^{\mathcal{Q}}_P(\bot)+\mathcal{A}^{\mathcal{Q}}_N(\bot)\right)$, which by the assertions of the observation is $o(1)$.
\end{proof}

We now define a specific super-oracle $\mathcal Q$ to be used with an input drawn by either $\mathcal{D}_P$ or $\mathcal{D}_N$ (defined above), and denote by $\mathcal{A}^{\mathcal{Q}}_P$ and $\mathcal{A}^{\mathcal{Q}}_N$ the respective distributions over transcripts when a deterministic algorithm $\mathcal A$ is run against $\mathcal Q$ with the input drawn by $\mathcal{D}_P$ and $\mathcal{D}_N$. We also define a super-oracle $\mathcal{Q}^-$ for the sake of some interim lemmas. 

\begin{definition}
	The super-oracle $\mathcal{Q}$ is defined as per Algorithm \ref{alg:supo}. The super-oracle $\mathcal{Q}^-$ is defined identically with the exception that Step 4 is not run (meaning that the oracle never performs a ``total reveal'' of the input).
\end{definition}

\begin{algorithm}
	\caption{Super-oracle $\mathcal{Q}$ for truestring equivalence}
	\label{alg:supo}
	Query provider for:  Input strings $w,w':\{1,\ldots,2n\}\to\{0,1\}$\\
	Persistent variables: $i_m\in\{0,\ldots,n\}$, initialized to $0$, and $Q,Q'\subseteq\{1,\ldots,n\}$, both initialized to $\emptyset$\\
	Activation parameter: A query $j$ from $w$ or a query $j'$ from $w'$
	\begin{enumerate}
		\item If a query $j\in\{1,\ldots,2n\}$ is requested from $w$, add $\min\{j,2n+1-j\}$ to $Q$
		\item If a query $j'\in\{1,\ldots,2n\}$ is requested from $w'$, add $\min\{j',2n+1-j'\}$ to $Q'$
		\item While there exists $j\in Q\cup Q'$ with $j\le i_m+100n^{4/5}$, set $i_m$ to $\max\{j,i_m\}$, set $Q$ to $Q\setminus \{1,\ldots,j\}$, and set $Q'$ to $Q'\setminus \{1,\ldots,j\}$
		\item If there exist $j\in Q$ and $j'\in Q'$ that are matching with respect to $w$ and $w'$, reveal the entirety of $w$ and $w'$ and terminate
		\item Reveal all values in $\{1,\ldots,i_m,2n+1-i_m,\ldots,2n\}\cup\bigcup\big\{\{j,2n+1-j\}:j\in Q\big\}$ from $w$ and  all values in $\{1,\ldots,i_m,2n+1-i_m,\ldots,2n\}\cup\bigcup\big\{\{j',2n+1-j'\}:j'\in Q'\big\}$ from $w'$
	\end{enumerate}
\end{algorithm}

The following observations are immediate.

\begin{observation}
	The super-oracle $\mathcal{Q}$ of Algorithm \ref{alg:supo} is oblivious (and so is $\mathcal{Q}^-$). In fact, outside the $\bot$ event, the revealed set only depends on the queried locations and not at all on the input (the~$\bot$ event itself as per Step 4 of the algorithm can depend on the entire input).
\end{observation}

\begin{observation}\label{obs:negcov}
	The distribution $\mathcal{A}^{\mathcal{Q}}_N$ underlies the distribution $\mathcal{A}^{\mathcal{Q}^-}_N$.
\end{observation}
\begin{proof}
    We note that the tree corresponding to $\mathcal{A}$ running with $\mathcal{Q}^-$ can be used as the tree for~$\mathcal{A}$ running with $\mathcal{Q}$, since the only difference is that some inputs can trigger $\bot$ events in some tree nodes with $\mathcal{Q}$ (but not with $\mathcal{Q}^-$). The observation then follows from a direct application of Observation~\ref{obs:condprob}.
\end{proof}

We next note that querying an index $i$ from $w$ will always reveal both $w[i]$ and $w[2n+1-i]$ (and similarly for querying an index from $w'$), and henceforth assume that all queries are made in the range $\{1,\ldots,n\}$. The following lemma bounds the increase of $i_m$ during $q$ interaction rounds.

\begin{lemma}\label{lem:imbound}
	After $q$ interaction rounds of a testing algorithm $\mathcal A$ against the super-oracle $\mathcal{Q}$ we have $i_m\leq q\cdot 100n^{4/5}$.
\end{lemma}
\begin{proof}
	Note that every round can add at most one index to $Q\cup Q'$. Any increase in $i_m$ involves the removal of at least one index from $Q\cup Q'$ (which could happen at a later round than the round where that index was inserted), and (by the condition of Step 3) each such removal increases $i_m$ by at most $100n^{4/5}$. The lemma follows.
\end{proof}

For given $U,U'\subset\{1,\ldots,2n\}$ which have a positive probability of being drawn as per the definition of $\mathcal{D}_P$ and $\mathcal{D}_N$, we define $\mathcal{A}^{\mathcal{Q}}_{P:U,U'}$ and $\mathcal{A}^{\mathcal{Q}}_{N:U,U'}$ as the respective distributions derived from $\mathcal{A}^{\mathcal{Q}}_P$ and $\mathcal{A}^{\mathcal{Q}}_N$ when we condition the drawing of the input on the specific $U$ and $U'$. We define $\mathcal{A}^{\mathcal{Q}^-}_{P:U,U'}$ and $\mathcal{A}^{\mathcal{Q}^-}_{N:U,U'}$ analogously. The following lemma is a crucial companion to Observation \ref{obs:negcov}.

\begin{lemma}\label{lem:pozcov}
	For any algorithm making $q=o(n^{1/5})$ queries and $n$ large enough, the distribution $\mathcal{A}^{\mathcal{Q}}_P$ also underlies the distribution $\mathcal{A}^{\mathcal{Q}^-}_N$.
\end{lemma}
\begin{proof}
	For $n$ large enough we have $q\cdot 100n^{4/5}<2n/5$, meaning by Lemma \ref{lem:imbound} that  $i_m$ (and $2n+1-i_m$) never reach a range in $w$ or $w'$ where the respective $u$ and $u'$ (as drawn in the definition of $\mathcal{D}_P$ and $\mathcal{D}_N$) differ. Additionally, $\mathcal{Q}^-$ never reveals the entire input. As with the proof of Observation \ref{obs:negcov}, we note that the tree corresponding to $\mathcal{A}$ running with $\mathcal{Q}^-$ can be used as the tree for $\mathcal{A}$ running with $\mathcal{Q}$. Let us now consider a leaf of this tree. It corresponds to the revealed set of its parent, calculated in Step 5 of Algorithm \ref{alg:supo} using $i_m$, $Q$ and $Q'$, and the corresponding values of the input,  namely $w|_{\{1,\ldots,i_m,2n+1-i_m,\ldots,2n\}\cup\bigcup\{\{j,2n+1-j\}:j\in Q\}}$ and $w'|_{\{1,\ldots,i_m,2n+1-i_m,\ldots,2n\}\cup\bigcup\{\{j',2n+1-j'\}:j'\in Q'\}}$.
	
	Recall that  an input drawn by $\mathcal{D}_P$ or $\mathcal{D}_N$ is a function of the sets $U,U'\subseteq \{1,\ldots,2n\}$ and the strings $u,u'\in\{0,1\}^n$. Recall also that $U$ and $U'$ are drawn in the same way and play the same role in the two distributions, and that additionally the total reveal behaviour of $\mathcal{Q}$ depends only on $U$ and $U'$ (which determine the true indexes in $w$ and $w'$) and the query history, and does not depend at all on $u$ or $u'$. Hence, to show for some leaf $a$ that $\mathcal{A}^{\mathcal{Q}}_P(a)\leq\mathcal{A}^{\mathcal{Q}^-}_N(a)$, it is enough to consider any $U$ and $U'$ that  are consistent with the revealed parts of $w$ and $w'$ and do not cause $\mathcal{Q}$ to reveal the entire input (by Algorithm \ref{alg:supo}, Step 4) on the respective root-to-leaf path, and show that  the corresponding conditional probabilities satisfy $\mathcal{A}^{\mathcal{Q}}_{P:U,U'}(a)\leq\mathcal{A}^{\mathcal{Q}^-}_{N:U,U'}(a)$.
	
	In fact, we show that the above two conditional probabilities are identical, by showing that the respective conditional distributions over $w|_{\{1,\ldots,i_m,2n+1-i_m,\ldots,2n\}\cup\bigcup\{\{j,2n+1-j\}:j\in Q\}}$ and $w'|_{\{1,\ldots,i_m,2n+1-i_m,\ldots,2n\}\cup\bigcup\{\{j',2n+1-j'\}:j'\in Q'\}}$ are identical (recall that in the decision tree, when it is guaranteed that the $\bot$ event cannot be triggered along the path, reaching the leaf $a$ is equivalent to having the respective values over the revealed sets).
	
	We now analyze the distribution over the revealed values of $w$ and $w'$ when $U$ and $U'$ are fixed as above, and $u$ and $u'$ are drawn from $\{0,1\}^n$ according to the definition of $\mathcal{D}_P$ and $\mathcal{D}_N$ (which as we recall, are independent of $U$ and $U'$). 
	Note that $u$ is drawn (as per the definition of both $\mathcal{D}_P$ and $\mathcal{D}_N$) uniformly from $\{0,1\}^n$, which (since $i_m\leq 2n/5$) together with $U$ and $U'$ determines all values of $w|_{\{1,\ldots,i_m,2n+1-i_m,\ldots,2n\}}$ and $w'|_{\{1,\ldots,i_m,2n+1-i_m,\ldots,2n\}}$ (recall that $u$ is identical to $u'$ on the relevant parts).
	
	Next, for every $j\in Q$ which matches some $1\leq i'\leq i_m$, again the values of $w$ at $j$ and $2n+1-j$ are completely determined by the already drawn sets ($U$ determines which of them is ``$\nochar$'' and which of them takes the non-$\nochar$ value among $w'[i']$ and $w'[2n+1-i']$, the latter value being equal to either $u'[t_{w'}(i')]=u[t_{w'}(i')]$ or $u'[t_{w'}(2n+1-i')]=u[t_{w'}(2n+1-i')]$, which were already treated above). An analogous consideration holds for any $j'\in Q'$ which matches some $1\leq i\leq i_m$. Denote by $\tilde{Q}$ and $\tilde{Q}'$ respectively the remaining indexes in $Q$ and $Q'$, namely those not matching any index between $1$ and $i_m$.
	
	Considering the values of $w$ and $w'$ on indices in $\tilde{Q}$ and $\tilde{Q}'$ (which are the only values of $w$ and $w'$ not treated above), note that by $\bot$ not occurring, none of them matches any other index in the revealed parts of $w$ and $w'$. This means that for any $j\in\tilde{Q}$, the value among $w[j]$ and $w[2n+1-j]$ which is not ``$\nochar$'' (as determined by $U$) is distributed uniformly over $\{0,1\}$ and independently of all other values that need to be revealed, and same goes for any $j'\in\tilde{Q}'$ not covered above with respect to the value among $w'[j']$ and $w'[2n+1-j']$ which is not ``$\nochar$'' (as determined by $U'$). This again holds over both $\mathcal{D}_P$ and $\mathcal{D}_N$.
	
	We have shown here that for every possible revealed set outside the total-reveal ``$\bot$'' event, the probability for any revealed input is the same for $\mathcal{D}_P$ and $\mathcal{D}_N$, which as discussed above in particular means that $\mathcal{A}^{\mathcal{Q}}_{P:U,U'}(a)\leq\mathcal{A}^{\mathcal{Q}^-}_{N:U,U'}(a)$, and hence (by considering all $U,U'$ that allow $\mathcal{A}$ to reach $a$ over $\mathcal{Q}$, which form a subset of the corresponding set over $\mathcal{Q}^-$)  we have $\mathcal{A}^{\mathcal{Q}}_P(a)\leq\mathcal{A}^{\mathcal{Q}^-}_N(a)$.
\end{proof}

It remains to show that the total reveal event $\bot$ happens with probability $o(1)$ as $n$ grows large. To give a good bound on two specific indexes from $Q$ and $Q'$ matching each other, we use the following technical lemma that follows from Stirling's formula.

\begin{lemma}\label{lem:scatter}
	If $X_1,\ldots,X_m$ are independent random variables with values chosen uniformly from $\{0,1\}$, then for $m$ large enough and all $j$ we have $\Pr [\sum_{i=1}^m X_i=j]\leq 1/\sqrt{m}$.
\end{lemma}

In our super-oracle (Algorithm \ref{alg:supo}) we make sure (by Step 3) that the $\bot$ event can only be triggered by an index that is not ``too close'' to the ``revealed range'' as it relates to $i_m$ (see Algorithm \ref{alg:supo}, Step 5).  Such indexes are subject to the following bound.

\begin{lemma}\label{lem:onebot}
	For $n$ large enough, given any $1\leq i\leq  n$, $i+100n^{4/5}< j\leq n$, and $Q'\subset\{1,\ldots,n\}$ with $\min Q'>i+100n^{4/5}$, the probability (under both $\mathcal{D}_P$ and $\mathcal{D}_N$) that $j$ matches any $j'\in Q'$ with respect to $w$ and $w'$, is bounded by $|Q'|/5n^{2/5}$, also when conditioned on the entirety  of $w$ and the values $w'|_{\{1,\ldots,i,2n+1-i,\ldots,2n\}\cup\bigcup\{\{j',2n+1-j'\}:j'\in Q'\}}$. The analogous statement holds for $i+100n^{4/5}< j'\leq n$ and $Q\subset\{1,\ldots,n\}$ with $\min Q>i+100n^{4/5}$ when conditioning on the entirety  of $w'$ and the values $w|_{\{1,\ldots,i,2n+1-i,\ldots,2n\}\cup\bigcup\{\{j,2n+1-j\}:j\in Q\}}$.
\end{lemma}
\begin{proof}
	We prove the lemma for the first statement, since the proof for the second statement is indeed analogous (swapping $w$ and $w'$). Fixing some $j'\in Q'$, since the event of $j$ matching $j'$ depends only on the identity of $U$ and $U'$ (which have the same distribution under both $\mathcal{D}_P$ and $\mathcal{D}_N$), the following analysis holds for both distributions. Recall that exactly one of $w[j]$ and $w[2n+1-j]$ is~``$\nochar$'', and the same goes for exactly one of $w'[j']$ and $w'[2n+1-j']$. We next prove that the probability that $j$ matches $j'$ is at most $1/5n^{2/5}$. To do so, we prove that the bound holds even if we additionally condition on any value of $w'|_{\{\min Q',\ldots,j'\}}$. The bound is proved by considering each of the following four cases.
	
	Consider first the case where $w[j],w'[j']\neq\nochar$. Let $X_k$ be the indicator variable for ``$w'[k+i]\neq\nochar$'' for $1 \le k \le m$, where $m=\min Q'-i-1\geq 100n^{4/5}-1>25n^{4/5}$. The locations $j$ and $j'$ match if and only if
    $$\sum_{k = 1}^{m} X_k = |\{1\leq l<j:w[l]\neq\nochar\}|-|\{1\leq l\leq i:w'[l]\neq\nochar\}|-|\{\min Q'\leq l<j':w'[l]\neq\nochar\}|$$
    By Lemma~\ref{lem:scatter}, the probability of this event is at most $1/\sqrt{25n^{4/5}}=1/5n^{2/5}$.
	
	In the case where $w[j]=w'[j']=\nochar$ (where $j$ and $j'$ can match due to the true indexes at $2n+1-j$ and $2n+1-j'$), we again use Lemma \ref{lem:scatter}, where each $X_k$ is the indicator variable for ``$w'[2n+1-k-i]\neq\nochar$'' for $1 \le k \le m$, where $m=\min Q'-i-1\geq 100n^{4/5}-1>25n^{4/5}$. We then have that $j,j'$ match if and only if
\begin{eqnarray*}
\sum_{k=1}^{m} X_k ~=~ \left|\left\{2n+1-j < l \leq 2n : w[l] \neq \nochar\right\}\right| & \!\!- & \!\!\left|\left\{2n+1-i \leq l \leq 2n : w'[l] \neq \nochar\right\}\right|\\
& \!\!- & \!\!\left|\left\{2n+1-j' < l \leq 2n+1-\min Q' :  w'[l] \neq \nochar\right\}\right|
\end{eqnarray*}
    By Lemma~\ref{lem:scatter}, the probability of this event is at most $1/\sqrt{25n^{4/5}}=1/5n^{2/5}$.
	
	In the final two cases, where exactly one of $w[j]$ and $w'[j']$ is ``$\nochar$'', no matching can occur. Thus the probability for a match between $j$ and $j'$ is bounded by $1/5n^{2/5}$. To finalize, we consider the union bound of this event for every $j'\in Q'$, to arrive at the $|Q'|/5n^{2/5}$ bound for the probability that $j$ matches any $j'\in Q'$.
\end{proof}

It is now easy to prove that the probability for $\bot$ under both $\mathcal{D}_P$ and $\mathcal{D}_N$ is small.

\begin{lemma}\label{lem:allbot}
	For any algorithm $\mathcal{A}$ making $q=o(n^{1/5})$ queries and the super-oracle $\mathcal{Q}$, we have $\mathcal{A}^{\mathcal{Q}}_P(\bot)=o(1)$ and $\mathcal{A}^{\mathcal{Q}}_N(\bot)=o(1)$.
\end{lemma}
\begin{proof}
	This result follows from a simple union bound using the per-round bound of Lemma \ref{lem:onebot}, noting that $q^2/5n^{2/5}=o(1)$.
\end{proof}

\begin{proof}[Proof of Lemma \ref{lem:undistinguishable} (implying Theorem \ref{th:TE_lower} and Corollary \ref{cor:D_m_lower})]
	Given any algorithm $\mathcal{A}$ making $q=o(n^{1/5})$ queries, by Lemma \ref{lem:allbot} the probability of $\bot$ (the total reveal event) under both distributions $\mathcal{A}^{\mathcal{Q}}_P$ and $\mathcal{A}^{\mathcal{Q}}_N$ is $o(1)$. This implies by Observation \ref{obs:negcov} and Observation \ref{obs:underdistance} that $d(\mathcal{A}^{\mathcal{Q}}_N,\mathcal{A}^{\mathcal{Q}^-}_N)=o(1)$, and by Lemma \ref{lem:pozcov} and Observation \ref{obs:underdistance} that $d(\mathcal{A}^{\mathcal{Q}}_P,\mathcal{A}^{\mathcal{Q}^-}_N)=o(1)$. Thus $d(\mathcal{A}^{\mathcal{Q}}_P,\mathcal{A}^{\mathcal{Q}}_N)=o(1)$.
	
	Putting the above together we have all the components to deduce by Observation \ref{obs:supcareful} the validity of Lemma \ref{lem:undistinguishable}.
\end{proof}

\section{Acknowledgements}
We thank Sofya Raskhodnikova, Diptaksho Palit, and Timothy Jackman for pointing out a flaw in our lower bound proof, which enabled us to correct it. 

\bibliographystyle{plain}
\bibliography{main}
\end{document}